\documentclass[12pt]{article}

\usepackage[australian]{babel}
\usepackage[margin=1in]{geometry}
\usepackage{setspace}
\usepackage{amsmath,amssymb,amsthm,mathtools}
\usepackage{booktabs}
\usepackage{graphicx}
\usepackage{hyperref}
\usepackage{natbib}
\usepackage[shortlabels]{enumitem}
\usepackage{tikz}
\usetikzlibrary{arrows.meta,positioning,fit,calc,shapes.geometric}

\usepackage{pgfplots}
\pgfplotsset{compat=1.18}

\newcommand{\E}{\text{E}}
\newcommand{\Bern}{\text{Bern}}
\newcommand{\KL}{\text{KL}}
\newcommand{\Var}{\text{Var}}

\hypersetup{
  colorlinks=true,
  linkcolor=blue,
  citecolor=blue,
  urlcolor=blue
}

\newtheorem{theorem}{Theorem}
\newtheorem{proposition}{Proposition}
\newtheorem{lemma}{Lemma}
\newtheorem{corollary}{Corollary}
\theoremstyle{definition}
\newtheorem{definition}{Definition}
\newtheorem{assumption}{Assumption}

\theoremstyle{remark}

\title{Optimal Use of Preferences in Artificial Intelligence Algorithms}
\author{Joshua S. Gans\thanks{Rotman School of Management, University of Toronto and NBER. Thanks to Philipp Strack for helpful comments, the SSHRC for financial assistance, and to Refine.ink, ChatGPT 5.2 and Claude Opus 4.5 for valuable research assistance. Responsibility for all errors remains my own.}}
\date{\today}

\begin{document}
\maketitle

\begin{abstract}
Machine learning systems embed preferences either in training losses or through post-processing of calibrated predictions. Applying information design methods from Strack and Yang (2024), this paper provides decision-problem-agnostic conditions under which separation—training preference-free and applying preferences ex post is optimal. Unlike prior work that requires specifying downstream objectives, the welfare results here apply uniformly across decision problems. The key primitive is a diminishing-value-of-information condition: relative to a fixed (normalised) preference-free loss, preference embedding makes informativeness less valuable at the margin, inducing a mean-preserving contraction of learned posteriors. Because the value of information is convex in beliefs, preference-free training weakly dominates for any expected-utility decision problem. This provides theoretical foundations for modular AI pipelines that learn calibrated probabilities and implement asymmetric costs through downstream decision rules. However, separation requires users to implement optimal decision rules. When cognitive constraints bind—as documented in human-AI decision-making—preference embedding can dominate by automating threshold computation. These results provide design guidance: preserve optionality through post-processing when objectives may shift; embed preferences when decision-stage frictions dominate.\textit{Journal of Economic Literature} Classification Numbers: C45, C53, D81, D82, D83.
\\

\end{abstract}

\newpage

\section{Introduction}\label{sec:intro}

Prediction is rarely an end in itself. Predictions are inputs into decisions: lending and insurance, admissions and hiring, medical triage, and the targeting of scarce public resources. This creates a fundamental architectural choice that appears in almost every applied pipeline, but is most salient in the deployment of artificial intelligence (AI). When building a prediction algorithm, should the training objective directly encode the eventual user's specific preferences (e.g., asymmetric costs or safety priorities), or should the algorithm be trained to learn the underlying truth (calibrated probabilities), leaving preferences to be implemented afterwards? This is not merely a technical question of hyperparameter tuning; it is a choice between building a specialised tool tied to a single, static objective or a general capability that preserves the option value of the information for diverse and changing downstream tasks.

A classical decision-theoretic view suggests that preferences belong in the loss function. In particular, \cite{grangerMachina2006} show how a decision problem induces an evaluation loss: the utility shortfall from acting on a forecast rather than on perfect information. When the forecasting method is taken as fixed, this logic is compelling: changing the evaluation loss changes how forecasts are ranked and how actions are chosen, but it does not change what information the forecast contains.\footnote{This approach has its roots in the work of \cite{varian1990goodness, zellner1986bayesian}; see \cite{gans2025} for a discussion in the context of AI and machine learning.}

To illustrate the main findings of the paper, it is useful to put this in the context of a simple statistical decision problem. Suppose that there is an unknown state of the world, $Y$, that can take on values of $0$ or $1$. There are two actions, $a \in \{0,1\}$, where the decision maker receives $0$ if they match the action to the state but face costs $c_{\mathrm{FP}}$ if $(a,Y) = (1,0)$ and $c_{\mathrm{FN}}$ if $(a,Y) = (0,1)$. Suppose the prior probability of $Y=1$ is $\frac12$. 

There are two ways to approach making this decision. Each involves generating a signal from available data with posterior, $q = \Pr[Y=1|\cdot]$. First, you could generate a signal that \textit{embeds preferences} in the signal itself. For instance, that signal could use class-weighted cross-entropy, $L^{w}(p,1)=-w_1\log p$ and $L^{w}(p,0)=-w_0\log(1-p)$, that generates a Bayes-optimal report at belief $q$ of:
\[
s(q)=\frac{w_1 q}{w_1 q+w_0(1-q)}.
\]
That is, the report is a cross-weighting of the posterior. In this case, the decision-maker will choose $a=1$ if and only if $s(q) \ge \frac12$ yielding an expected payoff of:
\[
\mathbb E[u_{\mathrm{emb}}]
=
-\mathbb E\!\left[\, q\,c_{\mathrm{FN}}\mathbf 1\!\left\{s(q)<\tfrac12\right\}
+(1-q)\,c_{\mathrm{FP}}\mathbf 1\!\left\{s(q)\ge\tfrac12\right\}\right].
\]
The second approach would be for the decision-maker to obtain a signal without any weights. The decision-maker then \textit{post-processes preferences}. The agent forms posterior belief $q$ and will optimally choose $a=1$ if and only if $q\ge \tau$, where
\[
\tau=\frac{c_{\mathrm{FP}}}{c_{\mathrm{FP}}+c_{\mathrm{FN}}}.
\]
This yields an expected payoff of:
\[
\mathbb E[u_{\mathrm{post}}]
=
-\mathbb E\!\left[\, q\,c_{\mathrm{FN}}\mathbf 1\{q<\tfrac{c_{\mathrm{FP}}}{c_{\mathrm{FP}}+c_{\mathrm{FN}}}\}+(1-q)c_{\mathrm{FP}}\mathbf 1\{q\ge\tfrac{c_{\mathrm{FP}}}{c_{\mathrm{FP}}+c_{\mathrm{FN}}}\}\right].
\]
To compare this with preference-embedding note that, since $s(q)\ge \tfrac12 \Longleftrightarrow q\ge \frac{w_0}{w_0+w_1}$, $\mathbb E[u_{\mathrm{emb}}]$ can be written as
\[
\mathbb E[u_{\mathrm{emb}}]
=
-\mathbb E\!\left[\, q\,c_{\mathrm{FN}}\mathbf 1\!\left\{q<\tfrac{w_0}{w_0+w_1}\right\}
+(1-q)\,c_{\mathrm{FP}}\mathbf 1\!\left\{q\ge\tfrac{w_0}{w_0+w_1}\right\}\right].
\]
Note that $\mathbb E[u_{\mathrm{post}}] = \mathbb E[u_{\mathrm{emb}}]$ if, following the \cite{grangerMachina2006} recommendation, $w_0:w_1=c_{\mathrm{FP}}:c_{\mathrm{FN}}$.

This equivalence highlights the challenge associated with preference-embedding. The problem is that, in practice, in machine learning, the loss is used not only to \emph{evaluate} predictions but to \emph{train} the algorithm that produces them. When learning is endogenous, changing the loss changes what the algorithm learns. This wedge is central in the cognitive-economic view of learning as endogenous information acquisition \citep{caplinMartinMarx2024}. Recently, \citet{autorCaplinMartinMarx2025} (ACMM) have sharpened this critique by identifying a specific incentive failure: they show that utility-weighted objectives ``flatten" the marginal incentives to learn in parts of the belief space, producing predictors that are, ironically, less useful even for the objective they were trained on. While they characterise this empirically as an incentive failure, operationally, training a preference-free probability model and then adjusting decisions ex post can outperform preference-embedded training even under the preference-weighted evaluation criterion.

Suppose, for instance, that preference embedding is implemented by training with a \emph{weighted} classification loss (e.g.\ class-weighted log loss) and then using a default cutoff on the model’s output (often $0.5$), on the theory that the weights have already encoded the preference asymmetry. The crucial difference from the frictionless illustrative above is that modern ML training almost never minimises the weighted loss alone: it minimises a weighted loss \emph{plus} standard controls on complexity such as weight decay, early stopping, or other regularisation. These controls penalise extreme decision boundaries and therefore shrink predicted probabilities toward the middle. Importantly, the regularisation term does \emph{not} scale with the cost weights in the same way as the data-fit term, so the mapping from “desired cost ratio” to “implemented decision rule” becomes unstable and model-dependent: changing the weights changes not just the implied cutoff but also what distinctions the model bothers to represent.

To see the consequence in the simplest possible way, suppose the underlying signal sometimes yields a low posterior $q_L$ and sometimes a high posterior $q_H$, with $q_L<\tau<q_H$. Under post-processing or \textit{separation}, a calibrated learner can preserve this spread and the decision-maker acts exactly in the high-posterior cases, as the threshold rule prescribes. Under weighted training with regularisation, the learned outputs are systematically pulled inward (the high cases become less extreme and the low cases become less low). In the extreme, the learner can collapse the two cases to nearly the same score, or at least eliminate the gap relative to the economically relevant threshold. Once that happens, no ex post thresholding can restore the lost option value: if the model output no longer separates the high-$q$ states from the low-$q$ states, then every downstream rule based only on that output must treat them (almost) the same. Separation avoids this failure mode because it targets the stable object (calibrated beliefs) and postpones the preference choice to a transparent, easily adjusted post-processing step.

This paper generalises these insights into a robust separation principle, framing the core architectural question as one of learning versus choosing: should the algorithm learn calibrated beliefs about the world, or should it directly learn which actions to choose? While ACMM demonstrate the risks of preference embedding for a fixed objective, this paper establishes conditions under which preference-free training (learning) is globally dominant across uncertain objectives. This is practically important because downstream objectives are often uncertain, contested, or changing: regulatory constraints evolve, safety policies are revised, organisational priorities shift, and heterogeneous users have differing needs.\footnote{This theme appears in \cite{agrawal2018prediction, agrawal2022prediction}, who distinguish the preference component in decision-making—\textit{judgment}, coming from the decision-maker—from the prediction component provided by AI.}

The approach models training as the choice of an information structure. Following Bayesian persuasion and information design \citep{kamenicaGentzkow2011} and drawing on the convex-order machinery of \citet{strackYang2024}, a trained predictor is summarised by the distribution of posterior beliefs it induces. The key primitive is a diminishing value of information condition: embedding preferences reduces the marginal gain from increasing informativeness. Combined with a one-dimensional comparability condition, this yields a clear comparative static: preference embedding induces a mean-preserving contraction of the learned posterior distribution (Theorem~\ref{thm:contraction}); a ``capability tax" that destroys information which cannot be recovered ex post.

The welfare consequence follows from the convexity of expected-utility decision problems in beliefs: more informative posteriors are weakly better, so the contraction result implies a separation principle (Theorem~\ref{thm:separation}). Training with a preference-free strictly proper scoring rule and implementing preferences only ex post weakly dominates preference-embedded training uniformly across expected-utility decision problems. This provides a decision-problem-agnostic foundation for modular AI design: separating knowledge acquisition from judgment application preserves the option value of the predictor.

Beyond the main separation result, the paper clarifies three further areas that matter for practice. First, ``preference-free'' is not a full specification of a training objective. The relevant class is that of (strictly) proper scoring rules, which elicit calibrated probabilities \citep{gneitingRaftery2007}. In binary settings, many strictly proper losses are available (locality does not pin down a unique one), so the choice among preference-free losses can depend on how Bayes-risk curvature interacts with the learning frictions. Section~\ref{sec:optimalLoss} develops a tractable ordering of preference-free losses and explains how to interpret ``optimality'' within a design class.

Second, it is possible to consider considerations that mitigate the preference for separation. The benchmark results assume that the decision maker can costlessly implement the optimal ex post decision rule given the learned probabilities. Section~\ref{sec:cognitive} introduces a decision-stage friction grounded in the rational inattention literature: using a rich probabilistic signal can itself be cognitively costly, and those costs scale with the information processed. In that case, embedding preferences upstream can be welfare-improving by compressing information and reducing decision-stage cognitive burden. Theorem~\ref{thm:sec9_RI_reversal} characterises when this ``reversal'' occurs.

Finally, Section~\ref{sec:rlhf} applies the same lens to a contemporary setting where ``preference embedding'' is literal: reinforcement learning from human feedback (RLHF) in large language model (LLM) deployment. This is the process by which LLM outputs are made honest, helpful and harmless.\footnote{For instance, OpenAI uses thousands of human reviewers to provide the feedback function; see \cite{stiennon2020learning, ouyang2022training}.} RLHF can be interpreted as an exponential tilting of an underlying posterior-quality distribution. When the objective is stable and aligned, RLHF improves welfare by shifting the level of quality (Theorem~\ref{thm:rlhf_fixed}). But when deployment objectives are uncertain or multi-dimensional, concentrating the generator on one fixed weighting reduces flexibility across objectives. This imposes a capability tax on foundation models (Theorem~\ref{thm:rlhf_separation}) and can amplify reward misspecification when the learned reward is misaligned with true quality (Proposition~\ref{prop:goodhart}). In this sense, the separation principle provides a rationale for modular pipelines that preserve optionality through post-processing when objectives may shift.

The remainder of the paper proceeds as follows. Section~\ref{sec:setup} introduces the environment and the distribution-over-posteriors representation. Section~\ref{sec:loss} reviews proper scoring rules and formalises the reporting/post-processing margin emphasised in ACMM. Section~\ref{sec:learning} models training as information acquisition with learning frictions and states the diminishing-value condition. Section~\ref{sec:mainresults} presents the contraction result and the robust separation principle. Section~\ref{sec:optimalLoss} discusses the choice among preference-free losses within common design restrictions and translates the theory into implementation guidance, worked examples, and diagnostics. Section~\ref{sec:cognitive} studies when cognitive constraints make preference embedding optimal. Section~\ref{sec:rlhf} analyses RLHF through the same embed-versus-post-process lens. Section~\ref{sec:conclusion} concludes.

\section{Model Setup and Preliminary Results}\label{sec:setup}

The core analysis focuses on a binary outcome to keep the convex-order comparisons clean and interpretable. Many AI prediction tasks can be cast in this scalar-belief form (for example, risk scoring). (Extensions beyond scalar beliefs require additional structure on information orders and are considered in Section \ref{sec:rlhf} and Appendix \ref{app:mv}.)

Let $Y \in \{0,1\}$ denote the outcome of interest. Let $\mu \in (0,1)$ denote the prior probability $\mathbb{P}(Y=1)=\mu$. A posterior belief about $Y$ is identified with a scalar $q \in [0,1]$, interpreted as $q=\mathbb{P}(Y=1\mid \cdot)$. A decision problem is a set of feasible actions $\mathcal{A}$ and a payoff function $u:\mathcal{A}\times\{0,1\}\to\mathbb{R}$. After observing information, the decision maker chooses an action to maximise expected payoff. Given belief $q$, the decision maker's indirect value is
\begin{equation}
V(q) \coloneqq \sup_{a\in\mathcal{A}} \left\{ q\,u(a,1) + (1-q)\,u(a,0) \right\}.
\end{equation}

\begin{lemma}[Convexity of indirect value]\label{lem:convexV}
For any decision problem $(\mathcal{A},u)$, the function $V:[0,1]\to\mathbb{R}$ is convex.
\end{lemma}

\begin{proof}
For each fixed action $a$, the map $q \mapsto q\,u(a,1)+(1-q)\,u(a,0)$ is affine. The pointwise supremum of affine functions is convex.
\end{proof}

\noindent Lemma \ref{lem:convexV} is the workhorse behind the welfare comparison. It formalises a simple point: when a decision maker can optimise actions after observing beliefs, the value of improving information is convex in beliefs. Lemma~\ref{lem:convexV} has an option-value interpretation. Information is valuable because it expands the set of contingent actions available after beliefs are realised. Convexity of $V$ implies that mean-preserving spreads in posterior beliefs raise expected indirect utility: dispersion in beliefs is a source of valuable flexibility, even before specifying any particular application.

A trained predictor is modelled as a signal about $Y$. Let $S$ denote the signal produced by the predictor. The distribution of posterior beliefs induced by $S$ is the distribution of the random variable
\begin{equation}
Q \coloneqq \mathbb{P}(Y=1\mid S).
\end{equation}
Bayes plausibility implies $\mathbb{E}[Q]=\mu$.

\begin{definition}[Bayes-plausible posterior distributions]
Let $\mathcal{Q}(\mu)$ denote the set of random variables $Q$ taking values in $[0,1]$ such that $\mathbb{E}[Q]=\mu$.
\end{definition}

\noindent This distribution-over-posteriors representation is standard in Bayesian persuasion and information design \citep{kamenicaGentzkow2011} and is used explicitly in \citet{strackYang2024}. It is also consistent with the emphasis in ACMM on separating what is learned (the induced posterior distribution) from what is reported (the predictor's output).

The relevant information order in the binary case is the convex order on $[0,1]$ random variables with fixed mean.

\begin{definition}[Convex order]\label{def:cx}
For $Q,Q' \in \mathcal{Q}(\mu)$, write $Q \succeq_{\mathrm{cx}} Q'$ if
\begin{equation}
\mathbb{E}[\varphi(Q)] \ge \mathbb{E}[\varphi(Q')]
\quad \text{for all convex } \varphi:[0,1]\to\mathbb{R}.
\end{equation}
Equivalently, $Q'$ is a mean-preserving contraction of $Q$.
\end{definition}

\noindent The interpretation is that $Q$ is more dispersed, hence more informative, than $Q'$ while preserving the same mean. This is the same order that underlies the mean-preserving contraction results in \citet{strackYang2024}.

A standard implication, combining Definition \ref{def:cx} with Lemma \ref{lem:convexV}, is the following.

\begin{lemma}[More informative posteriors improve expected value]\label{lem:infoValue}
If $Q \succeq_{\mathrm{cx}} Q'$, then for any decision problem $(\mathcal{A},u)$,
\begin{equation}
\mathbb{E}[V(Q)] \ge \mathbb{E}[V(Q')].
\end{equation}
\end{lemma}

\begin{proof}
By Lemma \ref{lem:convexV}, $V$ is convex, so Definition \ref{def:cx} implies $\mathbb{E}[V(Q)]\ge \mathbb{E}[V(Q')]$.
\end{proof}

\noindent Lemma \ref{lem:infoValue} is a Blackwell-type result specialised to binary beliefs. It will allow the paper to translate contraction results for posterior distributions into welfare comparisons for a broad class of decision problems.

\subsection{Loss Functions and Post-Processing}\label{sec:loss}

We now clarify what it means for a training loss to be ``preference-free'' in a rigorous sense, and it explains the ACMM-style decomposition between reporting and learning. 

A \textit{probabilistic prediction} is a reported probability $p \in [0,1]$, interpreted as the predictor's stated probability that $Y=1$. A \textit{training loss} is a function $L:[0,1]\times\{0,1\}\to\mathbb{R}$. Given belief $q$, the expected loss from reporting $p$ is
\begin{equation}
\bar L(p;q) \coloneqq q\,L(p,1) + (1-q)\,L(p,0).
\end{equation}

\begin{definition}[Strictly proper scoring rule]\label{def:proper}
A loss function $L$ is a strictly proper scoring rule if, for every $q\in[0,1]$, the unique minimiser of $\bar L(p;q)$ is $p=q$.
\end{definition}

\noindent Strict propriety captures the formal sense in which a loss is ``preference-free'' for learning probabilities: when the predictor's internal belief is $q$, the loss is minimised by truthfully reporting $q$.\footnote{See \citet{gneitingRaftery2007} for a comprehensive discussion.}

A key object associated with a loss is its Bayes risk.

\begin{definition}[Bayes risk]
Given a loss $L$, define the Bayes risk $H_L:[0,1]\to\mathbb{R}$ by
\begin{equation}
H_L(q) \coloneqq \min_{p\in[0,1]} \bar L(p;q).
\end{equation}
\end{definition}

\noindent For strictly proper scoring rules, the Bayes risk is attained at $p=q$, so $H_L(q)=\bar L(q;q)$. Proper scoring rules are in one-to-one correspondence with concave Bayes risk functions under standard regularity conditions \citep{gneitingRaftery2007}. This concavity is the channel through which proper losses reward informative posterior variation.

A common approach to embedding preferences is to modify the training loss to penalise some errors more than others. For binary classification, a canonical example is class-weighted cross-entropy with weights $w_1,w_0>0$:
\begin{equation}\label{eq:weightedCE}
L^{w}(p,1) = -w_1 \log(p),
\qquad
L^{w}(p,0) = -w_0 \log(1-p).
\end{equation}
This loss is not strictly proper with respect to the true posterior $q$ unless $w_1=w_0$, because the Bayes-optimal report is a transformation of $q$:
\begin{equation}\label{eq:bayesActWeighted}
p^{w}(q)=\frac{w_1 q}{w_1 q + w_0 (1-q)}.
\end{equation}
The transformation in \eqref{eq:bayesActWeighted} is exactly the kind of post-processing map highlighted in ACMM: conditional on belief $q$, preference embedding changes what is reported, but the reported object contains no more information than $q$ and may contain strictly less when inversion is not feasible or when reporting is compressed.\footnote{In multi-class settings, ACMM provide a general characterisation: for a strictly proper base loss, an optimal preference-weighted prediction corresponds to a deterministic transformation of the posterior based on expected utility components. The central message for the present paper is the same in binary and multi-class cases: conditional on the posterior, embedding preferences in the prediction is a post-processing step.}

\begin{proposition}[Bayes risk and curvature for class-weighted cross-entropy]
\label{prop:weightedCE_curvature}
For the class-weighted cross-entropy loss in \eqref{eq:weightedCE}, the Bayes risk is
\begin{equation}
\label{eq:weightedCE_bayesrisk}
H_{w}(q)
=
-w_1 q \log\!\big(w_1 q\big)
-w_0 (1-q)\log\!\big(w_0(1-q)\big)
+\big(w_1 q + w_0(1-q)\big)\log\!\big(w_1 q + w_0(1-q)\big),
\end{equation}
(with the usual convention $0\log 0:=0$). It is concave on $(0,1)$, with
\begin{equation}
\label{eq:weightedCE_secondderivative}
H_w''(q)
=
-\frac{w_0 w_1}{q(1-q)\big(w_0 + q(w_1-w_0)\big)}.
\end{equation}
If $H_{\log}$ denotes the Bayes risk of the unweighted log loss ($w_0=w_1=1$), then
\begin{equation}
\label{eq:weightedCE_increment_curvature}
\big(H_w-H_{\log}\big)''(q)
=
\frac{w_0 + q(w_1-w_0) - w_0 w_1}{q(1-q)\big(w_0 + q(w_1-w_0)\big)}.
\end{equation}
In particular, $\big(H_w-H_{\log}\big)$ is convex on $(0,1)$ if and only if $\max\{w_0,w_1\}\le 1$.
Thus, for common cost-sensitive normalisations in which one of $(w_0,w_1)$ exceeds $1$, the
diminishing-value condition in Assumption~\ref{ass:DV} must be checked rather than taken for granted.
\end{proposition}

\begin{proof}
The Bayes-optimal report is $p^w(q)$ in \eqref{eq:bayesActWeighted}, so
$H_w(q)=q\,L^{w}(p^{w}(q),1)+(1-q)\,L^{w}(p^{w}(q),0)$, which simplifies to \eqref{eq:weightedCE_bayesrisk}.
Differentiating twice yields \eqref{eq:weightedCE_secondderivative} and hence
\eqref{eq:weightedCE_increment_curvature}. Since the numerator in \eqref{eq:weightedCE_increment_curvature}
is affine in $q$, the convexity condition reduces to checking it at $q=0$ and $q=1$, giving
$\max\{w_0,w_1\}\le 1$.
\end{proof}

\noindent The normalisation here matters. The class weights $(w_0,w_1)$ are identified only up to a common scale: multiplying both weights by a constant rescales the loss without changing the Bayes-optimal report in \eqref{eq:bayesActWeighted}. In the learning model \eqref{eq:learningProblem}, however, such rescalings change learning incentives relative to the friction $C(\cdot)$. Accordingly, curvature comparisons like Proposition~\ref{prop:weightedCE_curvature} should be interpreted holding the overall scale of the loss fixed. Under alternative cost-sensitive normalisations with $\max\{w_0,w_1\}>1$, the diminishing-value condition in Assumption~\ref{ass:DV} can fail and preference embedding can in principle increase informativeness rather than contract it.

Any welfare difference between preference-embedded and preference-free pipelines must come from a difference in the posterior distribution induced by training, not from the existence of a reporting transformation. This is the motivation for the information acquisition model in the next section.

\subsection{Training and Learning Frictions}\label{sec:learning}

The paper models training as an endogenous choice of informativeness, consistent with the cost-based approach in \citet{caplinMartinMarx2024}. The model is intentionally reduced-form. It is not assumed that a particular optimisation algorithm solves a particular stochastic control problem. Instead, the aim is to capture a robust implication of many practical training pipelines: changing the loss changes the incentives to expend modelling capacity, optimisation effort, and representation power on learning informative distinctions.

\begin{assumption}[Cost-based learning as posterior choice]\label{ass:costBased}
For each training loss $L$, the trained algorithm induces a posterior random variable $Q_L \in \mathcal{Q}(\mu)$ that solves
\begin{equation}\label{eq:learningProblem}
Q_L \in \arg\min_{Q \in \mathcal{Q}(\mu)} \left\{ \mathbb{E}\!\left[ H_L(Q) \right] + C(Q) \right\},
\end{equation}
where $C:\mathcal{Q}(\mu)\to\mathbb{R}\cup\{+\infty\}$ is a learning friction function.
\end{assumption}

\noindent The interpretation is straightforward. More informative posteriors are beneficial because they reduce expected Bayes risk, and the loss $L$ determines how valuable those reductions are. The function $C(Q)$ captures optimisation frictions, regularisation, limited capacity, or any other features that make more informative predictors more costly to obtain. This is consistent with the pseudo-cost interpretation in \citet{caplinMartinMarx2024}.

\paragraph{Existence.} Before proceeding, we need a basic well-posedness result: for each training objective index $t$, the reduced-form “learning problem” in \eqref{eq:learningProblem} actually has a solution. Proposition~\ref{prop:existence} provides mild conditions under which an optimal (Bayes-plausible) posterior distribution exists, so later comparative-statics arguments are not about empty sets.

\begin{proposition}[Existence of an optimal posterior]\label{prop:existence}
Fix $t\in[0,1]$. Suppose the Bayes risk function $H_t$ is continuous on $[0,1]$. Suppose the learning friction $C$ is \emph{law-invariant} (it depends only on the distribution of $Q$), bounded below on $\mathcal{Q}(\mu)$, and lower semicontinuous under weak convergence of the induced distributions on $[0,1]$. Assume there exists at least one $Q\in\mathcal{Q}(\mu)$ with $C(Q)<\infty$. Then the minimisation problem in \eqref{eq:learningProblem} admits at least one minimiser.
\end{proposition}

\begin{proof}
Identify each $Q\in\mathcal{Q}(\mu)$ with its distribution $P_Q$ on $[0,1]$, and let $\mathcal{P}(\mu)$ denote the set of probability measures on $[0,1]$ with mean $\mu$. Because $[0,1]$ is compact and the mean constraint defines a closed set, $\mathcal{P}(\mu)$ is compact under weak convergence.

Define the objective on $\mathcal{P}(\mu)$ by
\[
F_t(P) \coloneqq \int_{0}^{1} H_t(q)\, dP(q) + \widetilde C(P),
\]
where $\widetilde C(P)$ is the common value of $C(Q)$ for any $Q$ with distribution $P$ (well-defined by law-invariance). Continuity of $H_t$ implies $P\mapsto \int H_t\, dP$ is continuous, and by assumption $\widetilde C$ is lower semicontinuous. Hence $F_t$ is lower semicontinuous on a compact set, so it attains a minimum. Any minimiser $P^\star$ can be realised as some $Q^\star\in\mathcal{Q}(\mu)$, yielding the claim.
\end{proof}

\noindent Proposition~\ref{prop:existence} shows the designer’s choice set is compact once we identify a posterior $Q$ with its induced law on $[0,1]$ and impose Bayes plausibility $\mathbb{E}[Q]=\mu$, and it says the objective is lower semicontinuous under weak convergence, so a minimiser must exist. With existence established, we can treat training outcomes as distributions over posteriors (an information-design move), because the optimisation can be carried out directly on $\mathcal{P}(\mu)$ without worrying that the infimum is only approached but never attained. This allows us to focus on how the \emph{shape} of $H_t$ and the \emph{order properties} of $C$ affect the optimal posterior(s). 

\paragraph{Learning Friction.} The next assumption expresses that producing more informative posteriors is, weakly, more costly.

\begin{assumption}[Convex-order monotone learning frictions]\label{ass:cxCost}
For $Q,Q' \in \mathcal{Q}(\mu)$, if $Q \succeq_{\mathrm{cx}} Q'$, then $C(Q) \ge C(Q')$.
\end{assumption}

\noindent Assumption \ref{ass:cxCost} is a learning-technology analogue of the idea that garbling a signal weakly reduces constraints or costs; see Figure \ref{fig:contraction}. It mirrors the Strack and Yang perspective that feasible objects under restrictions form a set closed under mean-preserving contractions \citep{strackYang2024}. Here, rather than their hard feasibility constraint, the restriction is encoded as a cost.

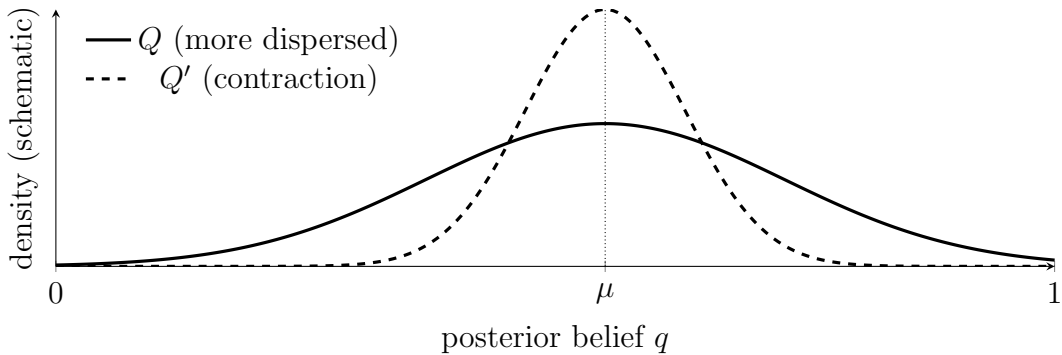
\begin{figure}[t]
\centering
\begin{tikzpicture}
\pgfmathsetmacro{\m}{0.55}   
\pgfmathsetmacro{\siga}{0.18} 
\pgfmathsetmacro{\sigb}{0.08} 

\begin{axis}[
  width=0.9\linewidth,
  height=5cm,
  xmin=0, xmax=1,
  ymin=0,
  axis lines=left,
  xlabel={posterior belief $q$},
  ylabel={density (schematic)},
  ytick=\empty,
  xtick={0,\m,1},
  xticklabels={0,$\mu$,1},
  domain=0:1,
  samples=250,
  legend style={draw=none, at={(0.02,0.98)}, anchor=north west}
]
\addplot[very thick] {exp(-((x-\m)^2)/(2*\siga^2))};
\addlegendentry{$Q$ (more dispersed)}

\addplot[very thick, dashed] {1.8*exp(-((x-\m)^2)/(2*\sigb^2))};
\addlegendentry{$Q'$ (contraction)}

\draw[densely dotted] (axis cs:\m,0) -- (axis cs:\m,2.2);
\end{axis}
\end{tikzpicture}
\caption{Illustration of $Q \succeq_{\mathrm{cx}} Q'$: $Q'$ is a mean-preserving contraction of $Q$ (schematic).}
\label{fig:contraction}
\end{figure}

Many plausible learning frictions satisfy Assumption \ref{ass:cxCost}.\footnote{$C(Q)$ captures learning frictions that plausibly arise in real-world machine learning models. For instance, regularisation techniques (such as $L_1$/$L_2$ weight decay or dropout) impose penalties on model complexity, effectively making highly dispersed posteriors more ``expensive'' to learn. This formulation also aligns with the Information Bottleneck principle, where the training objective explicitly penalises the mutual information retained about the input. Beyond explicit penalties, $C(Q)$ can represent implicit resource constraints---such as limited training data or early stopping budgets---that prevent the model from reaching fully informative (but potentially overfitted) configurations. While the assumption of law-invariance abstracts away architectural inductive biases (e.g., a convolutional neural network's structural preference for spatial features), the functional robustly captures the general trade-off between predictive accuracy and the costs of representation or optimisation.} A particularly clean microfoundation is information-theoretic capacity: producing a more informative signal is costly because it requires more ``bits'' of representation, optimisation effort, or data. Let $S$ denote the signal (representation) produced by the trained algorithm, and let $Q=\mathbb{P}(Y=1\mid S)$ be the induced posterior belief. Consider the mutual-information learning friction
\begin{equation}\label{eq:C_MI}
C_{\mathrm{MI}}(Q)\;:=\;\kappa\, I(Y;S), \qquad \kappa>0,
\end{equation}
as in rational inattention models \citep{sims2003}. In the binary case, this cost can be written purely as a functional of the posterior random variable $Q$:
\begin{equation}\label{eq:MI_as_Q}
I(Y;S)
=
h_{\mathrm{bin}}(\mu)\;-\;\mathbb{E}\!\left[h_{\mathrm{bin}}(Q)\right]
=
\mathbb{E}\!\left[\mathrm{KL}\!\left(\mathrm{Bern}(Q)\,\|\,\mathrm{Bern}(\mu)\right)\right],
\end{equation}
where $h_{\mathrm{bin}}(q)\coloneqq -q\log q-(1-q)\log(1-q)$ is binary entropy.
Here $\mathrm{Bern}(q)$ denotes the Bernoulli distribution on $\{0,1\}$ with success probability $q$
(i.e.\ $\mathbb{P}(1)=q$ and $\mathbb{P}(0)=1-q$), so $\mathrm{Bern}(Q)$ means the Bernoulli distribution
with success probability equal to the realised value of the posterior $Q$ (equivalently, $Y\mid S \sim \mathrm{Bern}(Q)$).
Finally, $\mathrm{KL}(P\|R)$ denotes the Kullback--Leibler divergence (relative entropy) from a distribution $P$ to $R$:
\[
\mathrm{KL}(P\|R)\;:=\;\sum_{y\in\{0,1\}} P(y)\log\!\left(\frac{P(y)}{R(y)}\right).
\]
In particular, for Bernoulli distributions,
\[
\mathrm{KL}\!\left(\mathrm{Bern}(q)\,\|\,\mathrm{Bern}(\mu)\right)
=
q\log\!\left(\frac{q}{\mu}\right) + (1-q)\log\!\left(\frac{1-q}{1-\mu}\right).
\]
Because $h_{\mathrm{bin}}$ is concave, $Q\succeq_{\mathrm{cx}}Q'$ implies $\mathbb{E}[h_{\mathrm{bin}}(Q)]\le \mathbb{E}[h_{\mathrm{bin}}(Q')]$, hence $C_{\mathrm{MI}}(Q)\ge C_{\mathrm{MI}}(Q')$. Therefore \eqref{eq:C_MI} satisfies Assumption~\ref{ass:cxCost}.%
\footnote{The same ``rate--distortion'' logic underlies information-theoretic regularisation and bottleneck-style objectives in ML. Related interpretations of $C(\cdot)$ include sample-size/data acquisition costs, capacity and bottleneck constraints (as in \cite{strackYang2024}), regularisation/MDL penalties, PAC--Bayes complexity bounds, and optimisation/early-stopping budgets.}

\paragraph{Key Condition.} The key step is to turn the idea ``preference embedding reduces information'' into a primitive condition. The model permits a continuous index of preference embedding, which allows a monotone comparative statics argument.

Let $t \in [0,1]$ index a family of training objectives. The case $t=0$ corresponds to preference-free training. The case $t=1$ corresponds to a preference-embedded variant, such as class-weighted cross-entropy or utility-weighted scoring. Formally, the index $t$ enters through the Bayes risk. Let $H_t$ denote the Bayes risk function associated with the $t$-indexed training objective.

\begin{assumption}[Diminishing value of information]\label{ass:DV}
The family $\{H_t\}_{t\in[0,1]}$ satisfies increasing differences with respect to convex order: for any $t_1>t_0$ and any $Q,Q' \in \mathcal{Q}(\mu)$ with $Q \succeq_{\mathrm{cx}} Q'$,
\begin{equation}\label{eq:increasingDiff}
\mathbb{E}\!\left[ H_{t_1}(Q) - H_{t_0}(Q) \right]
\ge
\mathbb{E}\!\left[ H_{t_1}(Q') - H_{t_0}(Q') \right].
\end{equation}
Moreover, the inequality in \eqref{eq:increasingDiff} is \emph{strict} whenever $Q \succeq_{\mathrm{cx}} Q'$ and $Q \not\stackrel{d}{=} Q'$.
\end{assumption}

\noindent Assumption \ref{ass:DV} can be read as a \emph{diminishing marginal value of informativeness} in terms of Bayes-risk reduction. For each $t$, define the (ex ante) value of an information structure $Q$ under the $t$-indexed training objective as the reduction in Bayes risk relative to receiving no information:
\begin{equation}\label{eq:voi}
\mathrm{VoI}_t(Q)\coloneqq H_t(\mu)-\mathbb{E}[H_t(Q)] .
\end{equation}
Because $H_t$ is concave in $q$, $\mathrm{VoI}_t(Q)\ge 0$, and it increases with informativeness in convex order. Assumption \ref{ass:DV} states that when $t$ increases (preferences are more heavily embedded), this value falls \emph{more} for more informative posteriors: for any $t_1>t_0$ and $Q\succeq_{\mathrm{cx}}Q'$,
\begin{equation}\label{eq:voiID}
\big(\mathrm{VoI}_{t_1}(Q)-\mathrm{VoI}_{t_0}(Q)\big)
\le
\big(\mathrm{VoI}_{t_1}(Q')-\mathrm{VoI}_{t_0}(Q')\big).
\end{equation}
Equivalently, for each $t_1>t_0$ the increment $\Delta_{t_1,t_0}(q)\coloneqq H_{t_1}(q)-H_{t_0}(q)$ is a convex function of $q$ (so $H_{t_1}$ is ``less concave'' than $H_{t_0}$). In this curvature sense, preference embedding flattens learning incentives.

In the binary setting, Assumption~\ref{ass:DV} reduces to a one-dimensional shape restriction.
Fix $t_1>t_0$ and define the increment $\Delta_{t_1,t_0}(q):=H_{t_1}(q)-H_{t_0}(q)$.
Assumption~\ref{ass:DV} holds if and only if $\Delta_{t_1,t_0}$ is convex on $[0,1]$
(and strictly convex for the strictness clause). When $H_t$ is twice differentiable, this is equivalent to checking $\Delta_{t_1,t_0}''(q)\ge 0$ for all $q$.
For many commonly used training objectives, $H_t$ is available in closed form; otherwise, the convexity check can be performed numerically on a fine grid.

This assumption is the precise primitive counterpart to the incentive-flattening intuition in ACMM. It is also the learning analogue of \cite{strackYang2024}'s result that privacy-preserving signals correspond to mean-preserving contractions of a benchmark posterior-mean distribution.

A sufficient condition for Assumption \ref{ass:DV} is particularly transparent.

\begin{proposition}[A convenient sufficient condition]\label{prop:sufficientDV}
If $H_t(q)=H_0(q)+t\,h(q)$ for some convex function $h:[0,1]\to\mathbb{R}$, then Assumption \ref{ass:DV} holds. If, in addition, $h$ is strictly convex, then the strictness clause in Assumption \ref{ass:DV} holds as well.
\end{proposition}

\begin{proof}
For $t_1>t_0$, $\mathbb{E}[H_{t_1}(Q)-H_{t_0}(Q)]=(t_1-t_0)\mathbb{E}[h(Q)]$. If $Q \succeq_{\mathrm{cx}} Q'$, then $\mathbb{E}[h(Q)]\ge \mathbb{E}[h(Q')]$ because $h$ is convex. This implies \eqref{eq:increasingDiff}. If $h$ is strictly convex and $Q \succeq_{\mathrm{cx}} Q'$ with $Q \not\stackrel{d}{=} Q'$, then the inequality is strict.
\end{proof}

\noindent Proposition \ref{prop:sufficientDV} shows that rather than assuming that preference embedding yields a contraction of learned posteriors, the proposition assumes a shape restriction on how preference embedding changes the training objective in the space of posterior distributions. The contraction emerges as a comparative statics implication.

A final assumption to operationalise our results is the following:

\begin{assumption}[Comparability of optimal posteriors]\label{ass:comparable}
For any $t_1>t_0$, and for any minimisers $Q_{t_1}\in\arg\min_{Q\in\mathcal{Q}(\mu)}\{\mathbb{E}[H_{t_1}(Q)]+C(Q)\}$ and $Q_{t_0}\in\arg\min_{Q\in\mathcal{Q}(\mu)}\{\mathbb{E}[H_{t_0}(Q)]+C(Q)\}$, the pair $(Q_{t_0},Q_{t_1})$ is comparable in convex order: either $Q_{t_0}\succeq_{\mathrm{cx}} Q_{t_1}$ or $Q_{t_1}\succeq_{\mathrm{cx}} Q_{t_0}$.
\end{assumption}

\noindent Assumption \ref{ass:comparable} is a substantive restriction. Convex order is only a partial order on $\mathcal{Q}(\mu)$, so in general two optimal posterior distributions need not be comparable. The assumption rules out these selection issues by requiring that optimisers move along a convex-order chain (a one-dimensional ``informativeness frontier''). It is used only to translate the inequality implied by Assumption \ref{ass:DV} into a monotone comparative-statics statement in convex order.\footnote{Assumption~\ref{ass:comparable} is best viewed as a \emph{single-index} restriction on the learning technology: optimisers move along a one-dimensional ``informativeness frontier.'' This is the natural case in many ML pipelines where learnability is governed by a small number of monotone knobs (sample size, training time, model size, regularisation strength, quantisation level). When a single knob tightens or relaxes a constraint, feasible predictors are nested by garbling or refinement, which induces a convex-order chain of posterior beliefs. The assumption is used only to obtain the convex-order contraction in Theorem~\ref{thm:contraction} from the increasing-differences inequality in Assumption~\ref{ass:DV}. Once a convex-order contraction between posteriors is established (by any means), the welfare step in Lemma~\ref{lem:infoValue} and hence the welfare comparison in Theorem~\ref{thm:separation} does not rely on comparability; we maintain Assumption~\ref{ass:comparable} in Theorem~\ref{thm:separation} only because its proof appeals to Theorem~\ref{thm:contraction}.}

Assumption~\ref{ass:comparable} is what turns the increasing-differences inequality in Assumption~\ref{ass:DV} into a global statement about \emph{all} minimisers. Without it, the argument behind Theorem~\ref{thm:contraction} still yields a weaker directional restriction: if a minimiser at $t_1$ is comparable in convex order to a minimiser at $t_0<t_1$, the only possible ordering is $Q_{t_0}\succeq_{\mathrm{cx}}Q_{t_1}$ (a higher-$t$ minimiser cannot be a strict mean-preserving spread of a lower-$t$ minimiser). Incomparability can arise when learning frictions are genuinely multi-dimensional, so the assumption should be read as a single-index restriction on the learning technology rather than a generic implication of arbitrary $C(\cdot)$. Unique minimisers and nested (Blackwell-ordered) signal families are two common cases in which comparability is automatically satisfied.

A sufficient condition for comparability is that the learning technology admits a nested family of attainable signals ordered by garbling (Blackwell order). In that case, the induced posterior means are ordered in convex order:

\begin{lemma}[Garbling implies convex-order contraction]\label{lem:garbling}
Let $S_0$ be a signal about $Y$ and let $S_1$ be a garbling of $S_0$ (that is, $S_1$ is generated from $S_0$ via a stochastic map independent of $Y$). Let $Q_i\coloneqq \mathbb{P}(Y=1\mid S_i)$. Then $Q_0 \succeq_{\mathrm{cx}} Q_1$.
\end{lemma}

\begin{proof}
Because $S_1$ is generated from $S_0$ via a stochastic map independent of $Y$, we have the Markov chain $Y \to S_0 \to S_1$ (equivalently, $Y\perp S_1\mid S_0$), so $\mathbb{E}[Y\mid S_0,S_1]=\mathbb{E}[Y\mid S_0]$. The tower property therefore gives
\[
Q_1=\mathbb{E}[Y\mid S_1]=\mathbb{E}\!\left[\mathbb{E}[Y\mid S_0]\mid S_1\right]=\mathbb{E}[Q_0\mid S_1].
\]
For any convex $\varphi$, conditional Jensen implies $\varphi(Q_1)=\varphi(\mathbb{E}[Q_0\mid S_1])\le \mathbb{E}[\varphi(Q_0)\mid S_1]$. Taking expectations yields $\mathbb{E}[\varphi(Q_1)]\le \mathbb{E}[\varphi(Q_0)]$, which is exactly $Q_0\succeq_{\mathrm{cx}} Q_1$ by Definition \ref{def:cx}.
\end{proof}

\section{Main Results: Contraction and Separation}\label{sec:mainresults}

The goal is to characterise whether embedding preferences can be optimal when post-processing can be conducted; see Figure \ref{fig:pipelines}. It will be demonstrated that it cannot under the assumptions just stated. We first establish a contraction result which, when combined with results already stated, produces a separation result.

\begin{figure}[t]
\centering
\begin{tikzpicture}[
  font=\small,
  node distance=5mm and 6mm,
  block/.style={draw, rounded corners, align=center, minimum height=7mm, minimum width=16mm},
  wide/.style={block, minimum width=26mm},
  arrow/.style={-{Latex[length=2mm]}, thick},
  frame/.style={draw, rounded corners, inner sep=3pt}
]
\node[block] (dataA) {Data};
\node[wide, right=of dataA] (trainA) {Training\\{\scriptsize pref.-weighted loss}};
\node[block, right=of trainA] (predA) {Prediction};
\node[block, right=of predA] (decA) {Decision};
\draw[arrow] (dataA) -- (trainA);
\draw[arrow] (trainA) -- (predA);
\draw[arrow] (predA) -- (decA);
\node[above=2mm of trainA, font=\bfseries\small] {Pipeline A: Embedded};
\node[frame, fit=(dataA)(trainA)(predA)(decA)] (frameA) {};

\begin{scope}[yshift=-20mm]
\node[block] (dataB) {Data};
\node[wide, right=of dataB] (trainB) {Training\\{\scriptsize pref.-free proper loss}};
\node[block, right=of trainB] (calB) {Calibrated\\$\hat q$};
\node[wide, right=of calB] (postB) {Post-processing\\{\scriptsize decision rule}};
\node[block, right=of postB] (decB) {Decision};
\draw[arrow] (dataB) -- (trainB);
\draw[arrow] (trainB) -- (calB);
\draw[arrow] (calB) -- (postB);
\draw[arrow] (postB) -- (decB);
\node[above=2mm of trainB, font=\bfseries\small] {Pipeline B: Separated};
\node[frame, fit=(dataB)(trainB)(calB)(postB)(decB)] (frameB) {};
\end{scope}

\node[align=center, font=\footnotesize, below=8mm of frameB] 
  {Under Theorem 2, separation weakly dominates (given Assumptions 1, 3, 4).};
\end{tikzpicture}
\caption{Two locations for preferences in a prediction pipeline: embed them in training (A) or apply them after learning calibrated probabilities (B).}
\label{fig:pipelines}
\end{figure}

\subsection{Preference embedding reduces learned informativeness}

The first theorem is a Strack-Yang-style contraction result. Increasing differences (Assumption~\ref{ass:DV}) pins down the direction of the comparative static whenever optimal posteriors are comparable, and Assumption~\ref{ass:comparable} ensures the optimiser correspondence lies on a convex-order chain so that this comparative static is well-defined.

\begin{theorem}[Preference embedding induces a mean-preserving contraction]\label{thm:contraction}
Suppose Assumptions \ref{ass:costBased}, \ref{ass:DV}, and \ref{ass:comparable} hold. For each $t\in[0,1]$, let $Q_t$ be a solution to \eqref{eq:learningProblem} with Bayes risk $H_t$.

If $t_1>t_0$, then every minimiser $Q_{t_0}$ at $t_0$ convex-order dominates every minimiser $Q_{t_1}$ at $t_1$:
\begin{equation}
Q_{t_0} \succeq_{\mathrm{cx}} Q_{t_1}.
\end{equation}
\end{theorem}

\begin{proof}
Fix $t_1>t_0$. Let $Q_{t_1}$ and $Q_{t_0}$ be minimisers at $t_1$ and $t_0$, respectively. Optimality implies
\begin{equation}\label{eq:opt1}
\mathbb{E}[H_{t_1}(Q_{t_1})] + C(Q_{t_1})
\le
\mathbb{E}[H_{t_1}(Q_{t_0})] + C(Q_{t_0}),
\end{equation}
and
\begin{equation}\label{eq:opt0}
\mathbb{E}[H_{t_0}(Q_{t_0})] + C(Q_{t_0})
\le
\mathbb{E}[H_{t_0}(Q_{t_1})] + C(Q_{t_1}).
\end{equation}
Adding \eqref{eq:opt1} and \eqref{eq:opt0} cancels the cost terms and yields
\begin{equation}\label{eq:keyIneq}
\mathbb{E}\!\left[H_{t_1}(Q_{t_1}) - H_{t_0}(Q_{t_1})\right]
\le
\mathbb{E}\!\left[H_{t_1}(Q_{t_0}) - H_{t_0}(Q_{t_0})\right].
\end{equation}

By Assumption \ref{ass:comparable}, either $Q_{t_0}\succeq_{\mathrm{cx}} Q_{t_1}$ (in which case we are done) or $Q_{t_1}\succeq_{\mathrm{cx}} Q_{t_0}$. Suppose $Q_{t_1}\succeq_{\mathrm{cx}} Q_{t_0}$. Then Assumption \ref{ass:DV} implies
\[
\mathbb{E}\!\left[H_{t_1}(Q_{t_1}) - H_{t_0}(Q_{t_1})\right]
\ge
\mathbb{E}\!\left[H_{t_1}(Q_{t_0}) - H_{t_0}(Q_{t_0})\right],
\]
with strict inequality whenever $Q_{t_1}\not\stackrel{d}{=}Q_{t_0}$. This contradicts \eqref{eq:keyIneq} unless $Q_{t_1}\stackrel{d}{=}Q_{t_0}$. In that knife-edge case, $Q_{t_0}\succeq_{\mathrm{cx}} Q_{t_1}$ holds as well. Hence $Q_{t_0}\succeq_{\mathrm{cx}} Q_{t_1}$ in all cases.
\end{proof}

\noindent Imposing, in addition, Assumption \ref{ass:cxCost} regarding learning frictions yields the following:

\begin{corollary}[Contraction weakly reduces learning cost]\label{cor:cost}
In addition to the assumptions of Theorem \ref{thm:contraction}, suppose Assumption \ref{ass:cxCost} holds. Then for any $t_1>t_0$ and any corresponding minimisers $Q_{t_0}$ and $Q_{t_1}$,
\[
C(Q_{t_0}) \ge C(Q_{t_1}).
\]
If, in addition, the inequality in Assumption \ref{ass:cxCost} is strict whenever $Q\succeq_{\mathrm{cx}}Q'$ and $Q\not\stackrel{d}{=}Q'$, then $C(Q_{t_0})>C(Q_{t_1})$ whenever $Q_{t_0}\not\stackrel{d}{=}Q_{t_1}$.
\end{corollary}

\begin{proof}
By Theorem \ref{thm:contraction}, $Q_{t_0} \succeq_{\mathrm{cx}} Q_{t_1}$. Assumption \ref{ass:cxCost} then implies $C(Q_{t_0}) \ge C(Q_{t_1})$. The strictness clause follows similarly.
\end{proof}

\noindent Theorem \ref{thm:contraction} is the formal counterpart to the ACMM-style claim that preference embedding weakens learning incentives and yields less informative predictors. It is also the direct analogue of Strack and Yang's mean-preserving contraction characterisation, with ``preference embedding'' playing the role of an additional restriction that induces a contraction \citep{strackYang2024}.

\subsection{Separation principle: learn probabilities, then decide}

The next theorem uses the contraction result to establish a robust welfare comparison.

\begin{theorem}[Robust separation principle]\label{thm:separation}
Suppose Assumptions \ref{ass:costBased}, \ref{ass:DV}, and \ref{ass:comparable} hold, and let $Q_0$ and $Q_1$ denote the posterior-belief random variables induced by preference-free and preference-embedded training, respectively.

Then, for any decision problem $(\mathcal{A},u)$ as defined in Section~\ref{sec:setup},
\begin{equation}
\mathbb{E}[V(Q_0)] \ge \mathbb{E}[V(Q_1)].
\end{equation}
Moreover, the welfare under the preference-free pipeline is achieved by applying preferences only ex post, through the Bayes-optimal action rule $a^\star(q)\in\arg\max_{a\in\mathcal{A}} \{q u(a,1)+(1-q)u(a,0)\}$.
\end{theorem}

\begin{proof}
By Theorem \ref{thm:contraction}, $Q_0 \succeq_{\mathrm{cx}} Q_1$. Lemma \ref{lem:infoValue} then implies $\mathbb{E}[V(Q_0)] \ge \mathbb{E}[V(Q_1)]$.

The second statement follows from the definition of $V(q)$ as the value of choosing a Bayes-optimal action given belief $q$.
\end{proof}

\noindent Theorem \ref{thm:separation} formalises a strong version of the practical guidance advocated by ACMM. Under the stated primitives, the optimal location for preferences is the ex post stage, not the training objective. The proof is deliberately transparent: contraction plus convexity implies dominance.

Theorem \ref{thm:separation} applies uniformly across expected-utility decision problems over the binary state $Y$ when the predictor's output is the decision maker's informational input about that state. This is broader than the binary-action classes in \citet{strackYang2024}, because the present model does not impose a privacy constraint that restricts the feasible set of signals. The reason such breadth is possible is that $V$ is convex for essentially all expected utility problems (Lemma \ref{lem:convexV}). If the decision maker also conditions on additional observables or constraints $Z$ at deployment, a uniform welfare comparison generally requires either a \emph{conditional} contraction $Q_0\mid Z\succeq_{\mathrm{cx}}Q_1\mid Z$ (so the result holds state-by-state in $Z$) or a stronger Blackwell comparison of the full signals; convex-order dominance of the marginal distribution of $Q$ alone need not be sufficient.

The binary outcome assumption remains important. In multi-dimensional belief settings, the convex order is partial, and selection issues become more pronounced. The paper's focus is consistent with many practical machine learning settings in economics where the core prediction object is a scalar risk score. In Appendix \ref{app:mv}, a multivariate version of Theorem \ref{thm:separation} is explored, demonstrating that some partial relaxation of the binary outcome assumption is possible.

\subsection{Closed-form illustration}\label{sec:solvedModel}

Theorems \ref{thm:contraction} and \ref{thm:separation} are developed for a general class of learning frictions. The next example shows the contraction mechanism in a simple, fully worked-out case where the optimal posterior distributions $Q_t$ can be written down in closed form.

Fix a prior $\mu\in(0,1)$. Consider the quadratic Bayes-risk family
\begin{equation}\label{eq:solvedHt}
H_t(q)\coloneqq q(1-q)+t\,q^{2}= q-(1-t)q^{2},\qquad t\in[0,1].
\end{equation}
This satisfies Assumption \ref{ass:DV} by Proposition \ref{prop:sufficientDV} with $h(q)=q^2$ (convex). Increasing $t$ makes $H_t$ less concave, capturing the idea that preference embedding flattens learning incentives.

Let learning frictions depend on posterior dispersion via the (law-invariant) functional
\begin{equation}\label{eq:solvedCost}
C(Q)\coloneqq \frac{\lambda}{2}\Big(\mathbb{E}[Q^{2}]-\mu^{2}\Big)^{2},\qquad \lambda>0.
\end{equation}
Because $q\mapsto q^{2}$ is convex and $\mathbb{E}[Q^{2}]\ge \mu^{2}$ for all $Q\in\mathcal{Q}(\mu)$, the cost \eqref{eq:solvedCost} is convex-order monotone (Assumption \ref{ass:cxCost}).

For any $Q\in\mathcal{Q}(\mu)$, write $x\coloneqq \mathbb{E}[Q^{2}]$. Since $\mathbb{E}[Q]=\mu$ is fixed and $H_t$ is quadratic,
\[
\mathbb{E}[H_t(Q)] = \mu-(1-t)\,x,
\]
so the learning problem \eqref{eq:learningProblem} reduces to a one-dimensional optimisation:
\begin{equation}\label{eq:solved1d}
\min_{x\in[\mu^{2},\,\mu]}\left\{\mu-(1-t)\,x+\frac{\lambda}{2}(x-\mu^{2})^{2}\right\}.
\end{equation}
The unique minimiser is
\begin{equation}\label{eq:solvedxstar}
x_t^\star = \Big(\mu^{2}+\frac{1-t}{\lambda}\Big)\wedge \mu,
\end{equation}
which is weakly decreasing in $t$.

To obtain an explicit optimal posterior distribution, pick the two-point posterior $Q_t$ that takes values in $\{0,q_t^{H}\}$ with mean $\mu$ and second moment $x_t^\star$:
\begin{equation}\label{eq:solvedQt}
Q_t=\begin{cases}
0 & \text{with probability } 1-\dfrac{\mu}{q_t^{H}},\\[0.7em]
q_t^{H} & \text{with probability } \dfrac{\mu}{q_t^{H}},
\end{cases}
\qquad
q_t^{H}\coloneqq \frac{x_t^\star}{\mu}\in[\mu,1].
\end{equation}
Then $\mathbb{E}[Q_t]=\mu$ and $\mathbb{E}[Q_t^{2}]=x_t^\star$, so $Q_t$ solves \eqref{eq:learningProblem}.

Finally, the family $\{Q_t\}_{t\in[0,1]}$ is totally ordered in convex order. For any $k\in[0,1]$,
\[
\mathbb{E}[(Q_t-k)_+] = 
\begin{cases}
\mu\left(1-\dfrac{k}{q_t^{H}}\right) & \text{if } k<q_t^{H},\\[0.7em]
0 & \text{if } k\ge q_t^{H},
\end{cases}
\]
which is pointwise increasing in $q_t^{H}$. Therefore, if $t_1>t_0$ then $q_{t_1}^{H}\le q_{t_0}^{H}$ and hence $Q_{t_0}\succeq_{\mathrm{cx}}Q_{t_1}$. Lemma \ref{lem:infoValue} then implies $\mathbb{E}[V(Q_{t_0})]\ge \mathbb{E}[V(Q_{t_1})]$ for every expected-utility decision problem, exactly as in Theorem \ref{thm:separation}.

This example illustrates a more general point. Start from the frictionless benchmark $C\equiv 0$ (or, in the example, $\lambda = 0$). Then training is just choosing a Bayes-plausible posterior distribution to minimise expected Bayes risk. Because Bayes risk is concave, the objective weakly favours more dispersed posteriors, so full revelation is an optimal solution (and it is the unique optimiser when Bayes risk is strictly concave; in the quadratic illustration, at $t=1$ the Bayes risk is linear and every Bayes-plausible $Q$ is optimal). In that benchmark, preference-embedded and preference-free training tie because learning is costless, so preferences only matter for the downstream action rule. Now introduce $C>0$, which captures that informativeness is costly due to capacity, optimisation, regularisation, or data constraints. Training becomes a genuine information-acquisition problem: you choose how dispersed posteriors should be by trading off Bayes-risk reduction against $C(Q)$. The key is that preference embedding changes the curvature of the Bayes risk, making it less concave in beliefs, so the marginal benefit of greater dispersion is lower. With the same learning friction, the optimiser therefore shifts to a mean-preserving contraction of posteriors under preference-embedded training. Since indirect utility is convex in beliefs, that contraction translates into a robust welfare loss, so the optimal place to apply preferences is ex post through the decision rule.

\subsection{A common decision environment}\label{sec:thresholdExample}

We continue with the binary statistical decision problem described in the Introduction. In that environment, the decision-maker observes data, forms a posterior belief $q=\Pr(Y=1\mid\cdot)$, and chooses $a\in\{0,1\}$ facing asymmetric misclassification costs $c_{\mathrm{FP}}$ and $c_{\mathrm{FN}}$. As shown there, the optimal policy is a threshold rule: act iff $q\ge \tau$, where $\tau=c_{\mathrm{FP}}/(c_{\mathrm{FP}}+c_{\mathrm{FN}})$. The Introduction also contrasts two pipelines: \emph{post-processing}, which first learns calibrated probabilities and then applies the threshold, and \emph{preference embedding}, which trains with a cost-sensitive loss (e.g.\ class-weighted cross-entropy) and then applies a default cutoff to the resulting score. In a frictionless benchmark, these two pipelines can be made decision-equivalent by choosing weights so that the induced score threshold implements the same $\tau$.

The point of this section is to explain why this equivalence can fail once learning is imperfect in the way it typically is in practice (finite data, regularisation, limited capacity, optimisation error). The mechanism is not that weights ``pick the wrong threshold''; the Introduction already makes clear how to align thresholds if one wishes. Rather, embedding changes the \emph{learning objective} and thereby changes the marginal value of informativeness. A convenient way to see this is through the curvature of the Bayes risk. Let $H^0(q)$ denote the Bayes risk of the unweighted log loss, and let $H^w(q)$ be the Bayes risk induced by the weighted loss in the Introduction. A direct calculation yields
\begin{equation}\label{eq:curvature_scaling_34}
H^{w\,\prime\prime}(q)=\rho(q)\,H^{0\,\prime\prime}(q),
\qquad
\rho(q):=\frac{w_0w_1}{(1-q)w_0+q w_1}.
\end{equation}
Under a standard scale normalisation of the weights, $\rho(q)\le 1$, so the embedded objective is less concave: additional dispersion in posteriors reduces expected Bayes risk by less. When informativeness is costly (for any of the reasons above), this weaker curvature induces a \emph{mean-preserving contraction} of the learned posterior distribution: compared to separated training, the predictor’s outputs are pulled toward the prior.

A two-point illustration makes the welfare implication transparent. Suppose separated training yields a calibrated posterior that takes values
\[
Q_0\in\{\mu-\delta,\ \mu+\delta\}\quad \text{with probability } \tfrac12 \text{ each},
\]
where $\mu=\mathbb E[Q_0]$ and $\delta>0$ indexes informativeness (so $\text{Var}(Q_0)=\delta^2$). Assume $\mu-\delta<\tau<\mu+\delta$, so the optimal threshold rule acts only when $Q_0=\mu+\delta$. Under preference embedding with learning frictions, let the learned posterior be a mean-preserving contraction
\[
Q_1\in\{\mu-\delta^w,\ \mu+\delta^w\}\quad \text{with probability } \tfrac12 \text{ each},
\qquad 0<\delta^w<\delta,
\]
so $\Var(Q_1)=(\delta^w)^2<\Var(Q_0)$. If the contraction is large enough that $\mu+\delta^w<\tau$, then the posterior produced by the embedded pipeline never crosses the economically relevant threshold, and the decision-maker optimally never acts. This is strictly suboptimal relative to separation because $\mu+\delta>\tau$ implies that acting is optimal in the high-posterior state under $Q_0$ but that state is no longer distinguishable under $Q_1$. Writing losses as negative payoffs, the expected loss under separation is
\[
\mathbb E[\ell_{\mathrm{sep}}]
=\tfrac12\Bigl[(\mu-\delta)c_{\mathrm{FN}}+\bigl(1-(\mu+\delta)\bigr)c_{\mathrm{FP}}\Bigr],
\]
while under embedding (when $\mu+\delta^w<\tau$) the decision-maker never acts, yielding
\[
\mathbb E[\ell_{\mathrm{emb}}]=\mathbb E[Q_1]\,c_{\mathrm{FN}}=\mu\,c_{\mathrm{FN}}.
\]
The welfare loss is therefore
\begin{equation}\label{eq:welfare_gap_34}
\mathbb E[\ell_{\mathrm{emb}}]-\mathbb E[\ell_{\mathrm{sep}}]
=\tfrac12\Bigl[(\mu+\delta)c_{\mathrm{FN}}-\bigl(1-(\mu+\delta)\bigr)c_{\mathrm{FP}}\Bigr]\;>\;0,
\end{equation}
where strict positivity follows from $\mu+\delta>\tau$ (equivalently, acting is optimal at the high posterior). Thus, even when the embedded objective is aligned with the decision-maker’s preferences in the frictionless benchmark, preference embedding can reduce welfare because it reduces informativeness in the learned signal in a way that cannot be repaired by ex post thresholding. Theorems~\ref{thm:contraction} and~\ref{thm:separation} formalise this intuition.

\section{Loss Choice and Implementation Under Separation}\label{sec:optimalLoss}\label{sec:implementation}

The separation principle answers \emph{where} preferences should enter the pipeline. It does not, by itself, identify a unique preference-free loss, because multiple strictly proper losses elicit calibrated probabilities and can interact differently with learning frictions.

This section clarifies what is and is not pinned down once one restricts attention to losses used in practice, and it makes precise the sense in which ``preference-free learning is optimal among variants'' should be interpreted. In this respect, it translates the formal results into a streamlined implementation message and illustrates the key mechanisms in a common example.

\subsection{Design class and the meaning of ``optimal among variants''}\label{sec:designClass}

Theorem~\ref{thm:separation} compares a preference-free family ($t=0$) to a preference-embedded family ($t=1$) holding fixed the underlying learning technology $C(\cdot)$. This corresponds closely to the practical question posed in ACMM: should asymmetric preferences be implemented by modifying the training loss or by post-processing a calibrated probability estimate?

To speak about optimality \emph{among variants}, one must specify a design class. In particular, if the designer is
free to scale the loss arbitrarily, a normalisation is required because scaling changes the trade-off between
expected Bayes risk and learning frictions in \eqref{eq:learningProblem}. In applied work, scaling also interacts
with optimisation choices (learning rates, early stopping, regularisation), so treating $C(\cdot)$ as fixed is most
naturally paired with treating the scale of admissible losses as fixed or normalised.

Two clarifications are useful. First, because posteriors are Bayes-plausible ($\mathbb{E}[Q]=\mu$), adding an affine function of $q$ to a Bayes risk does not affect the learning problem: if $\tilde H(q)=H(q)+a q+b$, then $\mathbb{E}[\tilde H(Q)]=\mathbb{E}[H(Q)]+a\mu+b$, so the set of minimisers of \eqref{eq:learningProblem} is unchanged. Thus, comparisons across losses are naturally understood \emph{up to affine terms} in the Bayes risk.
Second, multiplicative rescalings \emph{do} matter, and should be viewed as part of the design choice: multiplying the loss by a constant is equivalent to changing the effective learning-friction schedule in \eqref{eq:learningProblem}. The ordering results below should therefore be read as guidance within practically relevant classes in which the scale of candidate losses is fixed (for example, by convention, by optimisation hyperparameters, or by an explicit
normalisation).

\subsection{Choosing a preference-free baseline}\label{sec:properLossChoice}

In probabilistic classification, many commonly used losses are \textit{local} in the sense that the loss for an outcome depends only on the probability assigned to that realised outcome. In the binary setting, this restriction is vacuous: any loss $L(p,y)$ can be written in local form by setting $\ell_1(p):=L(p,1)$ and $\ell_0(s):=L(1-s,0)$. Thus, locality does not select a unique strictly proper scoring rule here.\footnote{In genuinely multi-class settings (with at least three outcome categories), locality becomes a substantive restriction and can single out the logarithmic score under standard regularity conditions; see \citet{gneitingRaftery2007} and references therein.}

The substantive content of separation in the binary setting is, therefore, about \emph{location} rather than
\emph{uniqueness}: train with a preference-free strictly proper loss to learn calibrated probabilities, then apply
preferences ex post. When multiple strictly proper losses are admissible, the choice among them can matter because
their Bayes-risk curvature governs the marginal benefit from generating more dispersed posteriors, and hence
interacts with the learning friction. Section~\ref{sec:curvatureOrder} formalises this curvature channel.

\subsection{A curvature order on preference-free losses}\label{sec:curvatureOrder}

When several proper losses are admissible, the model delivers a simple comparison principle: conditional on a scale
normalisation, prefer losses whose Bayes risk provides stronger incentives to generate informative posterior
dispersion. A convenient way to formalise this is an order on Bayes risk functions.

\begin{definition}[Concavity order on Bayes risk]\label{def:concavityOrder}
For two Bayes risk functions $H$ and $\tilde H$, write $H \preceq \tilde H$ if $\tilde H - H$ is convex.
\end{definition}

\noindent If $\tilde H-H$ is convex, then $\tilde H$ is ``less concave'' than $H$: relative to $H$, the objective
$\tilde H$ adds a convex penalty on dispersion. In the binary setting, this order is easy to verify. In particular,
if $H$ and $\tilde H$ are twice differentiable on $(0,1)$, then $H\preceq \tilde H$ is equivalent to
$\tilde H''(q)-H''(q)\ge 0$ for all $q\in(0,1)$.

The next proposition makes the ``optimal loss'' interpretation precise. It translates the curvature order in
Definition~\ref{def:concavityOrder} into an ordering of learned posterior distributions, under the same primitives
used in Section~\ref{sec:mainresults}.

\begin{proposition}[Curvature-ordered proper losses induce more informative learning]\label{prop:properLossOrder}
Fix two strictly proper scoring rules $L^0$ and $L^1$ with Bayes risks $H_0$ and $H_1$, and suppose $H_0 \preceq H_1$.
For $t\in[0,1]$ define the interpolated loss $L^t:=(1-t)L^0+tL^1$ and let $H_t$ denote its Bayes risk.
Maintain Assumption~\ref{ass:costBased}, and suppose Assumption~\ref{ass:comparable} holds for the family
$\{H_t\}_{t\in[0,1]}$. For each $t$, let $Q_t$ be any minimiser of \eqref{eq:learningProblem} with Bayes risk $H_t$.
Then if $t_1>t_0$,
\[
Q_{t_0}\succeq_{cx} Q_{t_1}.
\]
In particular, the posterior induced by training under $L^0$ convex-order dominates that induced by training under
$L^1$.

Moreover, for any decision problem $(\mathcal{A},u)$,
\[
\mathbb{E}\!\left[V(Q_0)\right]\ \ge\ \mathbb{E}\!\left[V(Q_1)\right].
\]
\end{proposition}

\begin{proof}
Because $L^t$ is a convex combination of strictly proper losses, it is strictly proper and its Bayes risk satisfies
$H_t=(1-t)H_0+tH_1=H_0+t(H_1-H_0)$. Since $H_1-H_0$ is convex, Proposition~\ref{prop:sufficientDV} implies that the family
$\{H_t\}_{t\in[0,1]}$ satisfies Assumption~\ref{ass:DV}. Theorem~\ref{thm:contraction} then yields
$Q_{t_0}\succeq_{cx} Q_{t_1}$ for any $t_1>t_0$. The welfare inequality follows from Lemma~\ref{lem:infoValue}.
\end{proof}

\noindent Proposition~\ref{prop:properLossOrder} provides a tractable notion of ``optimality among preference-free variants'': within a fixed design class, a loss with a more concave Bayes risk induces (weakly) more informative learned posteriors and therefore weakly dominates for \emph{all} expected-utility decision problems. The order $\preceq$ is a partial order, so not every pair of losses is comparable. But within many parametric families, or for common losses used in practice, checking the convexity of the increment is straightforward.

\subsection{Log loss versus Brier loss}\label{sec:logVsBrier} 

Here, Proposition~\ref{prop:properLossOrder} is illustrated using two standard strictly proper losses. The Bayes risk under log loss is Shannon entropy,
\begin{equation}
H_{\log}(q)=-q\log(q)-(1-q)\log(1-q),
\end{equation}
while the Bayes risk under the Brier loss is
\begin{equation}
H_{\mathrm{Brier}}(q)=q(1-q).
\end{equation}
A direct calculation shows that the increment $H_{\mathrm{Brier}}-H_{\log}$ is convex on $(0,1)$:
\[
\left(H_{\mathrm{Brier}}(q)-H_{\log}(q)\right)'' \;=\; \frac{1}{q(1-q)}-2 \;>\;0 \qquad \forall q\in(0,1).
\]
Therefore $H_{\log}\preceq H_{\mathrm{Brier}}$ in the sense of Definition~\ref{def:concavityOrder}. Under the conditions of Proposition~\ref{prop:properLossOrder} (in particular, holding fixed the same learning friction $C(\cdot)$ and assuming comparability along the interpolation), training under log loss induces a posterior distribution that convex-order dominates that induced by the Brier loss:
\[
Q_{\log}\ \succeq_{cx}\ Q_{\mathrm{Brier}}.
\]
By Lemma~\ref{lem:infoValue}, this implies a robust welfare comparison: for any downstream decision problem,
$\mathbb{E}[V(Q_{\log})]\ge \mathbb{E}[V(Q_{\mathrm{Brier}})]$.

The conclusion depends on holding fixed the learning technology in \eqref{eq:learningProblem}, including the effective scale of the loss. Rescaling a loss is equivalent to rescaling the learning friction and can overturn curvature comparisons. This is exactly why Section~\ref{sec:designClass} emphasises that ``optimal among variants'' requires a design class and a scale normalisation.

\subsection{Further thoughts}\label{sec:diagnosticsShort}

The contraction logic has an observable implication at the level of predicted probabilities. Under the conditions of Theorem~\ref{thm:contraction} (preference embedding) or Proposition~\ref{prop:properLossOrder} (switching to a less concave proper loss), the induced distribution of \emph{calibrated} predicted probabilities $\hat q(X)$ is predicted to undergo a mean-preserving contraction. In empirical work this can be implemented by calibrating (and thus aligning means across pipelines); in the strictly proper-loss comparison of Proposition~\ref{prop:properLossOrder} the means coincide by Bayes plausibility, so mean alignment is mainly a practical convenience. Equivalently, for a grid of cutoffs $k\in[0,1]$, one can compare the empirical functions $k\mapsto \mathbb{E}\big[(\hat q^{(t)}(X)-k)_+\big]$, whose pointwise ordering characterises convex order. This implication is falsifiable: systematic crossings of these curves indicate that the contraction mechanism is not the right description of that environment.

The present theory reinforces ACMM in a precise way. Their empirical findings suggest that heavy reweighting can reduce learning quality and that ex post adjustment can outperform preference-embedded training
\citep{caplinMartinMarx2024, autorCaplinMartinMarx2025}. Theorem~\ref{thm:contraction} rationalises this by showing that preference embedding can act like an effective penalty on informativeness, producing a mean-preserving contraction of the learned posterior distribution. Theorem~\ref{thm:separation} then converts this contraction into a welfare
comparison for a broad class of economic decision problems.

Proposition~\ref{prop:properLossOrder} makes an additional point that is implicit in much applied practice: even within the class of preference-free strictly proper losses, curvature matters for learning incentives. When two strictly proper losses differ by a convex increment in Bayes risk (holding fixed scale), the more concave Bayes risk induces more informative learned posteriors and therefore weakly dominates for any downstream decision problem.

At the same time, the results clarify scope. The paper does not claim that preference embedding is always harmful.
When the learning technology violates Assumption~\ref{ass:DV} or Assumption~\ref{ass:cxCost}, preference embedding can in principle be beneficial for specific objectives. These cases are important in practice, but they are distinct from the separation setting analysed here.

\section{When Preference Embedding Can Dominate}\label{sec:cognitive}

The separation principle in Theorem~\ref{thm:separation} is deliberately strong: it assumes that, once calibrated posteriors are available, decision makers can \emph{costlessly} interpret them and implement the Bayes-optimal action rule. This is the standard decision-theoretic benchmark, and it underwrites the canonical implementation advice: train to predict probabilities, then impose preferences ex post.

This section qualifies that benchmark by introducing a \emph{decision-stage} friction grounded in the rational inattention (RI) literature \citep{sims2003}. The key idea is simple: even if a predictive model delivers a valid posterior, \emph{using} that posterior requires scarce cognitive resources (attention, time, working memory, probabilistic sophistication). In RI models, cognitive effort scales with the \emph{information processed}. This makes cognitive costs endogenous to the stochastic environment: using an informative predictor is more demanding than using an uninformative one. Preference embedding can then become optimal because it effectively \emph{compresses} the predictive signal, reducing the cognitive burden on users, even though it sacrifices informativeness as in Theorem~\ref{thm:separation}.

The resulting trade-off is transparent: preference-free training yields more informative posteriors (higher decision value) but imposes higher attention cost; preference-embedded training yields less informative posteriors (lower decision value) but is cheaper to process.

\subsection{Motivating studies}

Empirical work documents systematic failures of post-processing in human--AI pipelines. \citet{agarwalMoehringRajpurkarSalz2023} study radiologists who receive AI-generated probability assessments for medical diagnoses. Although clinicians were informed that the AI scores were well-calibrated, many overrode confident AI predictions with their own less accurate judgments. In aggregate, the human--AI combination performed so poorly that the authors recommend allocating many cases either to humans without AI support or to the AI without human intervention.

This behaviour is a direct violation of the separation benchmark: calibrated probabilities are not automatically translated into optimal actions. A natural interpretation is that the effective cognitive cost of \emph{processing} probabilities is high: users may default to coarse heuristics, ignore probabilistic information, or apply it inconsistently under time pressure.

More broadly, the literature on algorithmic aversion documents that many users underweight or reject machine-generated advice even when it improves accuracy \citep{tejedaKumarSmythSteyvers2022}. Cognitive economics provides a theoretical foundation for such behaviour: using probabilistic information requires scarce cognitive inputs, and those inputs have opportunity costs \citep{caplin2025cognitive}. In this section, we formalise these observations using an RI-based cost of attention.

\subsection{Extension setup}

The extension here maintains the binary outcome environment from Section~\ref{sec:setup}. Let $V(q)$ denote the decision maker's indirect value function for the relevant downstream task. As established in Lemma~\ref{lem:convexV}, $V$ is convex in the belief $q$ for any expected-utility decision problem.

Let $Q_0$ and $Q_1$ denote the posterior-belief random variables induced by preference-free and preference-embedded training, respectively. The separation principle (Theorem~\ref{thm:separation}) implies that preference-free training yields weakly higher value absent decision-stage frictions. We define the \emph{information loss} from embedding as:
\begin{equation}\label{eq:sec9_info_loss_def}
\Delta_I:=\E[V(Q_0)]-\E[V(Q_1)]\ge 0.
\end{equation}
This quantity $\Delta_I$ captures the gross reduction in decision value caused by the contraction (in convex order) of the learned posterior beliefs, before accounting for any cognitive costs of processing that information.

\subsection{Rational inattention at the decision stage}

The cognitive-cost model in Section~\ref{sec:cognitive} is meant to capture a simple fact: extracting decision-relevant content from a probabilistic prediction requires cognitive effort. A natural discipline for such effort is the RI literature \citep{sims2003}, in which attention costs scale with the mutual information between the state and the signal processed.

Concretely, let $S$ denote the (processed) signal the decision maker attends to, and let $Q=\Pr(Y=1\mid S)$ be the induced posterior. Mutual information satisfies:
\begin{equation}\label{eq:sec9_MI_def}
I(Y;S)=H(Y)-H(Y\mid S)=h(\mu)-\E[h(Q)],
\end{equation}
where $\mu=\Pr(Y=1)$ and $h(q)=-q\log q-(1-q)\log(1-q)$ is the binary entropy function.\footnote{Logs can be taken in base $e$ (nats) or base $2$ (bits). This only rescales the attention-price parameter.}
Equivalently,
\begin{equation}\label{eq:sec9_MI_KL}
I(Y;S)=\E\!\left[\KL\!\big(\Bern(Q)\,\|\,\Bern(\mu)\big)\right].
\end{equation}
Because $Q$ is a sufficient statistic for $S$ about $Y$, one may write $I(Y;S)=I(Y;Q)$. We therefore model decision-stage cognitive effort as proportional to $I(Y;Q)$: the more the signal reduces uncertainty about $Y$ (relative to the prior), the more information must be absorbed to exploit it.

In standard RI, an agent can design an optimal information structure subject to an entropy cost. In human--AI settings, users typically face a \emph{menu} of deployed tools (interfaces) rather than the ability to redesign the signal at will. A radiologist does not get to redesign the model's UI; they decide how much to rely on it, whether to consult it, and whether to default to their baseline heuristic.

We capture this product environment with a menu-based RI assumption: the user can either (i) ignore the tool and act on the prior, or (ii) process the tool's signal as provided and then act optimally given the resulting posterior.

\begin{assumption}[Decision-stage RI cost with a limited menu]\label{ass:RI_menu_sec9}
Fix an attention-price parameter $\lambda_{\mathrm{cog}}>0$.
Under pipeline $i\in\{0,1\}$, the user chooses between:
\begin{enumerate}[label=(\alph*),leftmargin=2em]
\item \emph{Ignore the tool} and act on the prior, obtaining value $V(\mu)$; or
\item \emph{Use the tool}: process its output so as to obtain posterior $Q_i$, then choose the Bayes-optimal action $a^\star(Q_i)$, incurring cognitive cost $\lambda_{\mathrm{cog}}\cdot I(Y;Q_i)$.
\end{enumerate}
Users cannot costlessly redesign, garble, or reformat the tool's output beyond this menu.
\end{assumption}

\noindent Assumption~\ref{ass:RI_menu_sec9} is the minimal RI structure that (i) makes cognitive costs endogenous to the informativeness of the deployed prediction, and (ii) captures algorithmic aversion/limited uptake as an explicit outside option (ignore and act on the prior). The ``limited menu'' clause is not a technical trick; it reflects the institutional reality that signal redesign is an engineering choice made upstream (or is itself cognitively costly).

\subsection{The attention-informativeness trade-off}

Under pipeline $i\in\{0,1\}$, the user's welfare is:
\begin{equation}\label{eq:sec9_Wi_def}
W_i(\lambda_{\mathrm{cog}}):=\max\Big\{V(\mu),\ \E[V(Q_i)]-\lambda_{\mathrm{cog}}\,I(Y;Q_i)\Big\}.
\end{equation}
Define the \emph{attention savings} from embedding:
\begin{equation}\label{eq:sec9_MI_gap_def}
\Delta_{MI}:=I(Y;Q_0)-I(Y;Q_1).
\end{equation}
The next lemma links attention savings directly to the contraction result that drives the paper.

\begin{lemma}[Preference embedding reduces mutual information]\label{lem:sec9_MI_monotone}
If $Q_0\succeq_{\mathrm{cx}}Q_1$ (as in Theorem~\ref{thm:contraction}), then $I(Y;Q_0)\ge I(Y;Q_1)$, so $\Delta_{MI}\ge 0$.
\end{lemma}

\begin{proof}
By concavity of $h(\cdot)$, $Q_0\succeq_{\mathrm{cx}}Q_1$ implies $\E[h(Q_0)]\le \E[h(Q_1)]$.
Using \eqref{eq:sec9_MI_def}, $I(Y;Q)=h(\mu)-\E[h(Q)]$ is therefore larger for $Q_0$ than for $Q_1$.
\end{proof}

\noindent Lemma~\ref{lem:sec9_MI_monotone} makes the trade-off immediate. Preference embedding lowers expected decision value (by $\Delta_I\ge 0$) \emph{and} lowers the attention cost of using the tool (by $\lambda_{\mathrm{cog}}\Delta_{MI}\ge 0$). The question is which effect dominates.

\subsection{When embedding dominates under RI}

We are now in a position to characterise the main result from the RI extension.

\begin{theorem}[RI reversal condition]\label{thm:sec9_RI_reversal}
Suppose Assumption~\ref{ass:RI_menu_sec9} holds and $\Delta_{MI}>0$.

\begin{enumerate}[label=(\alph*),leftmargin=2em]
\item (\emph{When is a tool used at all?})
Pipeline $i$ is used (rather than ignored) if and only if
\begin{equation}\label{eq:sec9_use_threshold}
\lambda_{\mathrm{cog}}<\bar\lambda_i,
\qquad
\bar\lambda_i
:=
\begin{cases}
\dfrac{\E[V(Q_i)]-V(\mu)}{I(Y;Q_i)} & \text{if } I(Y;Q_i)>0,\\
0 & \text{if } I(Y;Q_i)=0.
\end{cases}
\end{equation}
When $I(Y;Q_i)=0$, Bayes plausibility implies $Q_i=\mu$ a.s., so the tool delivers $V(\mu)$ and is
never strictly preferred to ignoring.

\item (\emph{Embedding versus post-processing when both tools are used.})
If $\lambda_{\mathrm{cog}}<\min\{\bar\lambda_0,\bar\lambda_1\}$, then preference-embedded training yields higher welfare,
$W_1(\lambda_{\mathrm{cog}})>W_0(\lambda_{\mathrm{cog}})$, if and only if
\begin{equation}\label{eq:sec9_lambda_star}
\lambda_{\mathrm{cog}}>\lambda^\star
:=\frac{\Delta_I}{\Delta_{MI}}.
\end{equation}
Equivalently, embedding dominates exactly when the attention savings exceed the information loss:
\[
\lambda_{\mathrm{cog}}\Delta_{MI}>\Delta_I.
\]

\item (\emph{A pure uptake channel.})
If $\bar\lambda_1>\bar\lambda_0$, then for all $\lambda_{\mathrm{cog}}\in(\bar\lambda_0,\bar\lambda_1)$ the preference-free tool is ignored
while the embedded tool is used, implying $W_1(\lambda_{\mathrm{cog}})>W_0(\lambda_{\mathrm{cog}})$.
\end{enumerate}
\end{theorem}

\begin{proof}
Part (a) follows from \eqref{eq:sec9_Wi_def}: pipeline $i$ is used iff
\[
\E[V(Q_i)]-\lambda_{\mathrm{cog}}I(Y;Q_i)>V(\mu).
\]
If $I(Y;Q_i)=0$, then \eqref{eq:sec9_MI_def} implies $\E[h(Q_i)]=h(\mu)$, and strict concavity of $h$
yields $Q_i=\mu$ a.s.; hence the inequality fails (it holds as equality) and the tool is never
strictly used. If $I(Y;Q_i)>0$, rearranging yields \eqref{eq:sec9_use_threshold}.

For part (b), when both are used we compare net values:
\begin{align*}
    \E[V(Q_1)]&-\lambda_{\mathrm{cog}}I(Y;Q_1)>\E[V(Q_0)]-\lambda_{\mathrm{cog}}I(Y;Q_0)\\
&\iff \lambda_{\mathrm{cog}}(I(Y;Q_0)-I(Y;Q_1))>\E[V(Q_0)]-\E[V(Q_1)].
\end{align*}
Substitute $\Delta_{MI}$ and $\Delta_I$.

Part (c) is immediate from part (a).
\end{proof}

\noindent Theorem~\ref{thm:sec9_RI_reversal} delivers a clean ``value per bit'' interpretation. The ratio
\[
\lambda^\star=\frac{\Delta_I}{\Delta_{MI}}
\]
is the attention price at which the user is indifferent. The numerator $\Delta_I$ is how much decision value is lost by embedding (less informative posteriors). The denominator $\Delta_{MI}$ is how many bits/nats of state information the user is spared from processing. Embedding dominates when the marginal value of those extra bits is low relative to their cognitive price.

The main comparative statics are immediate from $\lambda^\star=\Delta_I/\Delta_{MI}$. First, as $\lambda_{\mathrm{cog}}$ rises (time pressure, multitasking, limited statistical training), the preference-free tool is both less likely to be \emph{used} (part (a)) and, conditional on use, less likely to be \emph{preferred} (part (b)). This provides a natural microfoundation for algorithmic aversion and low uptake: highly informative probability tools can be ignored because they are cognitively expensive to exploit.

Second, when learning frictions are severe, both $Q_0$ and $Q_1$ may be close to uninformative, and the information loss $\Delta_I$ is small. In that case, even modest attention savings can rationalise embedding. This aligns with the paper's broader message: when the achievable informativeness frontier is low, the opportunity cost of compressing predictions is low.

Third, holding fixed the information loss $\Delta_I$, preference embedding is more likely to dominate when the attention savings $\Delta_{MI}$ is large, since the threshold $\lambda^\star=\Delta_I/\Delta_{MI}$ falls. More generally, what matters for the reversal is the ratio $\Delta_I/\Delta_{MI}$: when changes to the upstream learning problem make the preference-free signal more informative, both $\Delta_I$ and $\Delta_{MI}$ typically move, so the net effect on $\lambda^\star$ is not signed without further restrictions.

Finally, individual differences in probabilistic sophistication can be represented as heterogeneity in $\lambda_{\mathrm{cog}}$. For instance, \citet{caplinDeanLeahy2024} document stable heterogeneity in calibration skill; in the present framework, such heterogeneity maps naturally to differences in the effective attention price, and therefore to differences in which pipeline users prefer.

\subsection{Why the limited-menu assumption matters}

Assumption~\ref{ass:RI_menu_sec9} restricts the user to a menu of deployed tools rather than allowing them to freely redesign the signal. This restriction is essential. If a decision maker could costlessly garble any received signal, then a more informative signal could never be worse: the user could always ignore or coarsen it optimally, replicating any outcome achievable under a less informative signal.

In particular, if $Q_1$ is a garbling of $Q_0$ (as in the contraction interpretation underlying Theorem~\ref{thm:separation}), then the feasible set of posteriors achievable from $Q_0$ contains that from $Q_1$, and an RI agent could weakly dominate the embedded pipeline by starting from $Q_0$ and (if desired) optimally compressing it. Thus, an embedding reversal is fundamentally a statement about \emph{rigidity in deployment}: users do not generally have access to costless, optimally designed garblings of model outputs; they face interfaces.

This observation also clarifies an important design implication: if a designer can keep preference-free training \emph{and} provide a cognitively light interface (for example, discretised risk buckets or explicit recommendations), then one can often obtain the benefits of attention savings without sacrificing training informativeness. Preference embedding is one way to shift cognitive burden upstream, but it is not the only way.

\section{RLHF as Preference Embedding}
\label{sec:rlhf}

This paper studies a simple design question: when should asymmetric preferences be built \emph{into}
the learning objective, and when they should be implemented \emph{after} learning through
post-processing and decision rules? The formal results are developed for scalar posteriors about
binary outcomes, but the same ``location of preferences'' question arises repeatedly in modern AI
systems.

This section answers two questions that connect the theory to current practice in large language
models (LLMs).  
\begin{enumerate}
    \item \textbf{How does the model relate to LLMs?} Many deployed LLM systems contain scalar--binary components that are \emph{exactly} of the form studied in Sections~\ref{sec:setup}--\ref{sec:loss}: safety and moderation classifiers, quality or acceptability scorers, and routing rules that map a score into a discrete action (respond, refuse, regenerate, or escalate). For these modules, the mapping is literal, and the separation principle has a direct implementation meaning: learn calibrated probabilities with a preference-free proper loss, then implement asymmetric trade-offs via downstream thresholds, routing, or best-of-$n$ selection.
    \item \textbf{Where does RLHF sit in these results?} Reinforcement learning from human feedback (RLHF) is a generator-side alignment method that modifies the \emph{output distribution} of an LLM so that outputs scored as ``preferred'' become more likely. In the present lens, RLHF is a form of preference embedding: it shifts the burden of implementing a particular weighting of downstream criteria upstream into training. This can be beneficial when the deployment objective is stable, and the learned reward is well-specified (Theorem~\ref{thm:rlhf_fixed}). But it can impose a capability tax when objectives vary, are multi-dimensional, or drift over time (Theorem~\ref{thm:rlhf_separation}), and it can amplify reward misspecification when the learned reward is misaligned with true quality (Proposition~\ref{prop:goodhart}). The practical distinction is the same as in the main text:
\emph{embedding} trades robustness and optionality for lower inference-time costs.
\end{enumerate}
Answering these questions involves the same design question but with different formal primitives. Theorems~\ref{thm:contraction}--\ref{thm:separation} compare posterior-belief random variables that share a common mean and are ordered by convex order. The RLHF results instead study how training changes a generator's output distribution through exponential tilting, which generally shifts means and is naturally compared using first-order stochastic dominance for a fixed objective. The common thread is therefore not the specific information order, but the location of preferences: baking an objective into training can reduce inference-time search or selection costs, but can reduce robustness when objectives drift or are misspecified.

The discussion proceeds as follows. Section~\ref{subsec:llm_binary} records the strict mapping for binary acceptability scoring. The remainder treats RLHF as an ``embed versus post-process'' choice for generators, emphasising the minimum technical machinery needed to interpret RLHF through the paper's results.

\subsection{Binary acceptability scoring in LLM deployment}
\label{subsec:llm_binary}

While the paper’s formal results are stated for scalar posteriors about binary outcomes, many LLM
deployments contain exactly such scalar--binary modules. Fix a prompt $X$ and a candidate completion
$Z$. Let $Y\in\{0,1\}$ where $Y=1$ if the completion is acceptable under the deployment policy, and $Y=0$ otherwise. A scoring model observes features of $(X,Z)$ and produces a signal $S$, inducing the posterior belief $Q=\mathbb{P}(Y=1\mid S)\in[0,1]$. A deployment policy then chooses an action $a\in\mathcal{A}$ (show $Z$, refuse, regenerate, ask a follow-up, or route to human review) with payoff $u(a,Y)$. This is exactly the environment of Section~\ref{sec:setup}. Table~\ref{tab:llm_mapping} records the correspondence.

\begin{table}[t]
\centering
\begin{tabular}{p{5.4cm} p{8.6cm}}
\toprule
\textbf{Paper object} & \textbf{LLM-system analogue (binary acceptability setting)} \\
\midrule
Outcome $Y\in\{0,1\}$ & Whether a candidate completion is acceptable (safe, policy-compliant, high-quality) \\
\addlinespace
Signal $S$ and posterior $Q$ & Score produced by a safety/reward/quality model and the implied acceptability probability \\
\addlinespace
Preference-free proper loss & Train the scoring model with a strictly proper loss to estimate $q=\mathbb{P}(Y=1\mid \cdot)$ \\
\addlinespace
Preference-embedded loss & Train the scoring model with an asymmetric loss (e.g.\ class-weighted cross-entropy) that hard-codes downstream trade-offs \\
\addlinespace
Decision problem $(\mathcal{A},u)$ & Deployment policy: show/refuse/regenerate/escalate; payoffs reflect safety and user value \\
\addlinespace
Post-processing / decision rule $a^\star(\cdot)$ & Thresholding, routing, best-of-$n$ selection, rejection based on the estimated $q$ \\
\bottomrule
\end{tabular}
\caption{Mapping the paper's framework to a common LLM deployment sub-problem: scoring candidate outputs for binary acceptability.}
\label{tab:llm_mapping}
\end{table}

In this setting, the separation message is straightforward: train the scorer with a preference-free
strictly proper loss so that it estimates $\mathbb{P}(Y=1\mid\cdot)$ (and calibrate if needed), then
implement asymmetric preferences through the downstream policy $a^\star(\cdot)$. In particular,
changes in deployment trade-offs (e.g.\ tighter safety standards or different latency costs) belong
naturally to the \emph{post-processing} and decision layer, not to re-weighting the scorer’s training
loss.

This strict application is intentionally narrow: it applies to the ubiquitous \emph{scoring and
routing} modules inside LLM systems. The remainder of this section turns to a different object:
end-to-end procedures, such as RLHF, that modify the generator itself.

\subsection{An economic primer on RLHF}
\label{subsec:rlhf_primer}

An LLM can be represented as a conditional distribution $\pi(z\mid x)$ over completions $z$ given a
prompt $x$. \emph{Pretraining} produces a baseline generator $\pi_0$ by maximum likelihood on large
text corpora. The result is a capable but undirected distribution over outputs: it reflects patterns
in its training data rather than any particular deployment objective.

RLHF is a widely used method for shifting $\pi_0$ toward outputs that humans judge as more desirable.
It is useful to separate two ingredients.
\begin{enumerate}
    \item \textbf{Reward modelling.} Humans compare pairs of candidate outputs for the same prompt and indicate which is preferable. A \emph{reward model} $r(x,z)$ is trained to predict these preferences, producing a scalar score that summarises ``how good'' an output looks according to the training preference data.\footnote{See \cite{ziegler2019finetuning, ouyang2022training, rafailov2023direct} for canonical treatments.}
    \item \textbf{Policy optimisation with a KL constraint.} The generator is then updated to increase expected reward while remaining close to $\pi_0$:
\begin{equation}
\max_{\pi} \; \mathbb{E}_{x \sim P(x)} \, \mathbb{E}_{z \sim \pi(\cdot\mid x)} \!\left[ r(x,z) \right]
\;-\; \lambda \cdot \KL(\pi \,\|\, \pi_0),
\label{eq:rlhf_objective}
\end{equation}
where $\KL(\pi \,\|\, \pi_0)=\mathbb{E}_{x}\!\left[\KL(\pi(\cdot\mid x)\,\|\,\pi_0(\cdot\mid x))\right]$ is the Kullback--Leibler divergence (relative entropy) from a distribution $\pi_0$ to $\pi$, and $\lambda>0$ controls the strength of the constraint.
\end{enumerate}
For economists, \eqref{eq:rlhf_objective} is a familiar object: it is a regularised expected-utility
problem over \emph{choice probabilities}. The KL term is an adjustment cost (or information cost)
penalising deviations from a baseline distribution. As in logit models, the parameter $\lambda$
governs how concentrated the choice probabilities become: small $\lambda$ corresponds to aggressive
optimisation of the score, while large $\lambda$ keeps the policy closer to the baseline.

The key design question is then immediate. One can implement preferences \emph{at training time} by solving \eqref{eq:rlhf_objective} (RLHF), or implement them \emph{at deployment time} by sampling candidates from $\pi_0$ and selecting among them using a scorer (reranking, rejection sampling, best-of-$n$, or routing). The former embeds a particular reward into the generator; the latter keeps the generator broadly capable and applies preferences ex post, in the spirit of separation.

\subsection{RLHF as preference embedding}
\label{subsec:rlhf_tilting}

To connect RLHF to the paper's objects, interpret each completion $z$ for prompt $x$ as having
a (latent) \emph{acceptability} probability $q(x,z)\in[0,1]$ for the relevant downstream notion of
``good'' output.\footnote{Appendix~\ref{app:rlhf_technical} provides a minimal microfoundation that
generates such a scalar $q(x,z)$ from a binary acceptability variable.}
For any generator $\pi$ and prompt $x$, let $Q_{\pi,x}$ denote the law of $q(x,Z)$ when
$Z\sim \pi(\cdot\mid x)$.
Under a prompt distribution $P(x)$, define the unconditional posterior-quality distribution as the mixture
\[
Q_\pi(\cdot):=\int Q_{\pi,x}(\cdot)\,dP(x).
\]
Equivalently, if $X\sim P$ and $Z\sim \pi(\cdot\mid X)$, then $Q_\pi$ is the law of $q(X,Z)$.
This reduces the generator comparison to a comparison of distributions on $[0,1]$, just as in the
main text.

\begin{assumption}[Aligned reward model]
\label{ass:aligned_reward}
The reward model is (cardinally) aligned with acceptability: $r(x,z)=q(x,z)$ for all $(x,z)$.
\end{assumption}

\noindent Assumption~\ref{ass:aligned_reward} isolates the mechanical effect of RLHF, abstracting from reward
misspecification. Under this assumption, RLHF has a sharp characterisation.

\begin{proposition}[RLHF as exponential tilting]
\label{prop:rlhf_optimal}
Under Assumption~\ref{ass:aligned_reward}, the solution to \eqref{eq:rlhf_objective} satisfies, for
each prompt $x$,
\begin{equation}
\pi_R(z\mid x)=\frac{\pi_0(z\mid x)\exp(q(x,z)/\lambda)}{Z_\lambda(x)},
\label{eq:rlhf_solution}
\end{equation}
where $Z_\lambda(x)=\int \pi_0(\tilde z\mid x)\exp(q(x,\tilde z)/\lambda)\,d\tilde z$.
Moreover, for each prompt $x$, the induced distribution $Q_{R,x}$ is obtained from $Q_{0,x}$ by the reweighting identity
\begin{equation}
\mathbb{E}_{Q_{R,x}}[\varphi(Q)]
=\frac{\mathbb{E}_{Q_{0,x}}\!\left[\varphi(Q)\exp(Q/\lambda)\right]}{\mathbb{E}_{Q_{0,x}}\!\left[\exp(Q/\lambda)\right]}
\qquad \text{for all bounded measurable }\varphi:[0,1]\to\mathbb{R}.
\label{eq:posterior_tilting}
\end{equation}
Consequently, unconditional expectations under $Q_R$ satisfy
\[
\mathbb{E}_{Q_R}[\varphi(Q)]
=
\int
\frac{\mathbb{E}_{Q_{0,x}}\!\left[\varphi(Q)\exp(Q/\lambda)\right]}{\mathbb{E}_{Q_{0,x}}\!\left[\exp(Q/\lambda)\right]}
\, dP(x).
\]
\end{proposition}

\noindent Proposition~\ref{prop:rlhf_optimal} shows that RLHF is a \emph{tilt} of the baseline policy
toward higher-acceptability outputs. The parameter $\lambda$ plays exactly the role of a logit scale:
as $\lambda\downarrow 0$ the policy concentrates on the highest-$q$ completions, while large $\lambda$
keeps the generator close to $\pi_0$. A proof is standard for entropy-regularised problems and is
recorded in Appendix~\ref{app:rlhf_technical} for completeness.

A simple implication is that RLHF increases the likelihood of high-$q$ outputs.

\begin{lemma}[A dominance property of RLHF]
\label{lem:tilting_properties}
Under Assumption~\ref{ass:aligned_reward}, for each prompt $x$ the tilted posterior distribution
$Q_{R,x}$ first-order stochastically dominates $Q_{0,x}$. Consequently, the unconditional mixture
$Q_R$ first-order stochastically dominates $Q_0$. In particular, $\mathbb{E}[Q_R]\ge \mathbb{E}[Q_0]$,
with strict inequality unless $Q_{0,x}$ is degenerate for $P$-almost every $x$.
\end{lemma}

\begin{proof}
Fix $x$. Equation \eqref{eq:posterior_tilting} shows that $Q_{R,x}$ is obtained from $Q_{0,x}$ by
exponential tilting with likelihood ratio proportional to $\exp(Q/\lambda)$, which is strictly
increasing in $Q$. This monotone-likelihood-ratio property implies that $Q_{R,x}$ first-order
stochastically dominates $Q_{0,x}$ (a direct proof is given in Appendix~\ref{app:rlhf_technical} using
Chebyshev's covariance inequality).

Since $Q_R$ and $Q_0$ are mixtures of $\{Q_{R,x}\}_x$ and $\{Q_{0,x}\}_x$ under the same prompt
distribution $P(x)$, first-order stochastic dominance is preserved under mixing: for each threshold
$t\in[0,1]$,
\[
\Pr_{Q_R}(Q\ge t)=\int \Pr_{Q_{R,x}}(Q\ge t)\,dP(x)
\ \ge\
\int \Pr_{Q_{0,x}}(Q\ge t)\,dP(x)=\Pr_{Q_0}(Q\ge t).
\]
The mean inequality follows from FOSD. Strictness fails only if the tilt is $P$-a.s.\ trivial, i.e.,
$Q_{0,x}$ is degenerate for $P$-almost every $x$.
\end{proof}

\noindent Lemma~\ref{lem:tilting_properties} already answers the basic ``is RLHF preference embedding?''
question in the present framework. Compared to a separated pipeline that keeps $\pi_0$ fixed and then
selects among its draws, RLHF changes the \emph{generator} so that a particular reward-weighting is
implemented upstream. The remaining question is when this change is desirable, and when it creates
the same kinds of robustness losses highlighted by Theorem~\ref{thm:separation}.

\subsection{When does RLHF improve welfare?}
\label{subsec:rlhf_separation}

The main welfare effect of RLHF is not a Blackwell comparison holding means fixed (as in the main
theorems), because RLHF shifts the \emph{level} of acceptability. When the deployment objective is
stable and aligned with the reward, this mean shift is beneficial.

\begin{theorem}[RLHF with fixed objectives]
\label{thm:rlhf_fixed}
Maintain Assumption~\ref{ass:aligned_reward}. Suppose the deployment objective can be written as
\[
W(\pi)=\mathbb{E}[V(Q_\pi)]
\]
for some nondecreasing $V:[0,1]\to\mathbb{R}$. Then, for any $\lambda>0$,
\[
W(\pi_R)\ge W(\pi_0).
\]
\end{theorem}

\begin{proof}
By Lemma~\ref{lem:tilting_properties}, $Q_R$ first-order stochastically dominates $Q_0$. Since $V$ is
nondecreasing, $\mathbb{E}[V(Q_R)]\ge \mathbb{E}[V(Q_0)]$.
\end{proof}

\noindent Theorem~\ref{thm:rlhf_fixed} clarifies why RLHF is often effective in practice: when the target is
well-defined and stable (``follow instructions under a fixed policy''), shifting the generator so
that higher-reward outputs occur more often directly improves expected performance.

\subsection{When does separation dominate? }

The separation logic becomes relevant when the relevant objective is not a single fixed scalar, but
a family of criteria whose weights can vary across contexts or over time (helpfulness, harmlessness,
truthfulness, style, domain constraints, regulatory rules, and so on). In that case, embedding one
fixed scalarisation into the generator can reduce flexibility under alternative scalarisations.

Formally, suppose each completion has a vector of criterion-wise acceptability posteriors
$q(x,z)\in[0,1]^K$ and that scalar objectives are linear functionals $r_w(x,z)=w\cdot q(x,z)$ for
weights $w\in\Delta^{K-1}$. RLHF trained on a single $w$ concentrates the generator on completions
that are good for that scalarisation; a separated design keeps a richer menu and chooses $w$-specific
actions ex post.

\begin{theorem}[Separation under objective uncertainty]
\label{thm:rlhf_separation}
Fix a prompt $x$ and suppose the baseline generator $\pi_0(\cdot\mid x)$ has finite support
$\mathcal{Z}(x)$. For each completion $z\in\mathcal{Z}(x)$ let $q(z)\in[0,1]^K$ and for any weight
vector $w\in\Delta^{K-1}$ define $r_w(z):=w\cdot q(z)$. Let $w$ denote the training weights used for
RLHF, and define the induced RLHF policy
\[
\pi_R^\lambda(z\mid x)=\frac{\pi_0(z\mid x)\exp(r_w(z)/\lambda)}
{\sum_{z'\in\mathcal{Z}(x)}\pi_0(z'\mid x)\exp(r_w(z')/\lambda)}.
\]
Let the deployment weights $W$ be random with support $\mathcal{W}\subset\Delta^{K-1}$ and suppose
$\mathbb{P}(W=w')>0$ for some $w'\neq w$. Assume there exist completions $z^w,z^{w'}\in\mathcal{Z}(x)$
such that $z^w$ is a strict maximiser of $r_w$ and $z^{w'}$ is a strict maximiser of $r_{w'}$, with
$r_{w'}(z^{w'})>r_{w'}(z^w)$.

Consider an idealised separation design that, after observing $W$, selects
$z^\star(W)\in\arg\max_{z\in\mathcal{Z}(x)} r_W(z)$. Then there exists $\bar\lambda>0$ such that for
all $\lambda<\bar\lambda$,
\[
\mathbb{E}_{W}\big[r_W(z^\star(W))\big]
>
\mathbb{E}_{W}\Big[\mathbb{E}_{Z\sim\pi_R^\lambda(\cdot\mid x)}\big[r_W(Z)\big]\Big].
\]
\end{theorem}

\begin{proof}
Because $\mathcal{Z}(x)$ is finite and $z^w$ is a strict maximiser of $r_w$, for every $z\neq z^w$,
\[
\frac{\pi_R^\lambda(z\mid x)}{\pi_R^\lambda(z^w\mid x)}
=\frac{\pi_0(z\mid x)}{\pi_0(z^w\mid x)}
\exp\!\left(\frac{r_w(z)-r_w(z^w)}{\lambda}\right)\to 0
\quad\text{as }\lambda\downarrow 0.
\]
Hence $\pi_R^\lambda(\cdot\mid x)\Rightarrow \delta_{z^w}$, so
$\mathbb{E}_{Z\sim\pi_R^\lambda(\cdot\mid x)}[r_{w'}(Z)] \to r_{w'}(z^w)$.
Since $z^{w'}$ is a strict maximiser of $r_{w'}$ and $r_{w'}(z^{w'})>r_{w'}(z^w)$, for $\lambda$
sufficiently small,
\[
\mathbb{E}_{Z\sim\pi_R^\lambda(\cdot\mid x)}[r_{w'}(Z)] < r_{w'}(z^\star(w')).
\]
Taking expectations over $W$ and using $\mathbb{P}(W=w')>0$ yields the strict inequality.
\end{proof}

\noindent Theorem~\ref{thm:rlhf_separation} formalises a simple idea: with objective uncertainty, embedding one
particular scalarisation into the generator, purchases performance on that scalarisation by giving up
menu flexibility under other scalarisations. This is the generator analogue of the welfare loss from
embedding a particular asymmetric loss into a predictor rather than learning calibrated probabilities
and choosing ex post.

Theorem~\ref{thm:rlhf_separation} shows that KL-regularised RLHF implements an
\emph{exponential tilt} of the base generator toward completions with higher training reward
$r_w$. When the regularisation parameter $\lambda$ is small, this tilt becomes sharply
concentrated: the generator behaves like an \emph{amortised optimiser} for $r_w$ and places
nearly all probability on (one of) the $r_w$-maximising completions.
This is beneficial when evaluation uses the same objective $r_w$, but it comes with a loss of
\emph{option value}: by concentrating on what is best for $w$, the generator becomes less
useful for alternative objectives $w'$ that may be relevant at deployment.

Theorem~\ref{thm:rlhf_separation} is stated for \emph{objective uncertainty} (random deployment
weights $W$) and compares an embedded RLHF generator to an idealised separated design that chooses the
best completion \emph{after} observing the objective. The following result isolates the
simplest and most common special case: deployment evaluates completions using a \emph{single} fixed
alternative objective $r_{w'}$ that differs from the training objective $r_w$. The quantity
\[
\sup_{z\in\mathcal{Z}(x)} r_{w'}(z)\;-\;\mathbb{E}_{Z\sim\pi_R^\lambda(\cdot\mid x)}[r_{w'}(Z)]
\]
can be read as an \emph{option-value loss} from embedding $w$ into the generator: the first term is
what a system could achieve if it retained a full ``menu'' of completions and then selected the best
one for $w'$ ex post, while the second term is what the RLHF-trained generator delivers when judged
by $w'$.

\begin{corollary}[Capability tax for a fixed alternative objective]
\label{cor:capability_tax}
In the setting of Theorem~\ref{thm:rlhf_separation}, fix $w'\neq w$ and suppose
$r_{w'}(z^{w'})>r_{w'}(z^w)$ for strict maximisers $z^w\in\arg\max r_w$ and
$z^{w'}\in\arg\max r_{w'}$. Then for all $\lambda$ sufficiently small,
\[
\sup_{z\in\mathcal{Z}(x)} r_{w'}(z)
\;-\;
\mathbb{E}_{Z\sim\pi_R^\lambda(\cdot\mid x)}[r_{w'}(Z)]
\;>\;0,
\]
and the gap vanishes only if $r_w$ and $r_{w'}$ share the same maximisers.
\end{corollary}

\begin{proof}
Apply Theorem~\ref{thm:rlhf_separation} with $W\equiv w'$.
\end{proof}

\noindent For small $\lambda$, the KL-regularised RLHF policy $\pi_R^\lambda$ concentrates its mass on
$r_w$-maximising completions (here captured by the strict maximiser $z^w$). If $w'$ ranks a different
completion $z^{w'}$ strictly above $z^w$, i.e.\ $r_{w'}(z^{w'})>r_{w'}(z^w)$, then the best achievable
value under the alternative objective is $\sup_{z} r_{w'}(z)=r_{w'}(z^{w'})$, whereas the RLHF
generator's expected $r_{w'}$-value is pulled toward $r_{w'}(z^w)$ as $\lambda$ becomes small. This
creates a strictly positive gap for sufficiently small $\lambda$. The only way this ``capability tax''
can vanish (in the small-$\lambda$ regime) is if the objectives are \emph{aligned at the top}, meaning
$r_w$ and $r_{w'}$ share the same maximisers, so that concentrating on $r_w$-optimal completions does
not exclude any $r_{w'}$-optimal completion.-

In many RLHF deployments, ``reward'' is not a primitive scalar but a \emph{scalarisation} of several
criteria (helpfulness, harmlessness, truthfulness, style, domain constraints, etc.), represented
here by $q^1,\ldots,q^K$. Training chooses one set of weights $w$ to turn these criteria into a single
objective $r_w$. But at deployment, the relevant trade-offs can differ across users, contexts, or time,
corresponding to a different weighting $w'$.The next result translates the
capability-tax idea into this multi-criterion language: embedding one weighting into the generator
buys performance on that weighting at the cost of flexibility under alternative weightings.

\begin{corollary}[Capability tax across criteria]
\label{cor:capability_tax_weights}
Suppose $q(z)=(q^1(z),\ldots,q^K(z))$ and RLHF is trained on $r_w(z)=\sum_{k} w_k q^k(z)$ for weights
$w\in\Delta^{K-1}$. Fix an alternative weighting $w'\neq w$. If the maximisers of $r_w$ and $r_{w'}$
differ, then for $\lambda$ sufficiently small the RLHF generator is strictly suboptimal under the
alternative scalarisation $r_{w'}$ relative to a design that selects completions using $w'$ ex post.
\end{corollary}

\begin{proof}
This is a direct restatement of Corollary~\ref{cor:capability_tax} with $r_{w'}(z)=\sum_k w'_k q^k(z)$.
\end{proof}

\noindent If the maximisers of $r_w$ and $r_{w'}$ differ, then there exists a completion that is optimal for the
alternative trade-off $w'$ but not selected by optimising for $w$. When $\lambda$ is small, the RLHF
generator concentrates on completions that are optimal for $w$, making it strictly suboptimal when
evaluated under $w'$. By contrast, a separated design that defers the scalarisation choice to
deployment can select a $w'$-maximiser ex post and therefore achieves $\sup_z r_{w'}(z)$. Formally, this
is exactly Corollary~\ref{cor:capability_tax} instantiated with the linear scalarisation
$r_{w'}(z)=\sum_k w'_k q^k(z)$.

\subsection{Reward misspecification}
\label{subsec:misspecification}

The analysis so far assumed the reward model is aligned with the true acceptability posterior
$q(x,z)$. In practice, $r(x,z)$ is estimated from finite preference data and is therefore an
imperfect proxy. When the proxy becomes the optimisation target, misspecification is amplified in
the familiar Goodhart (``when a measure becomes a target, it stops being a good measure") sense: selection pressure pushes the generator toward regions of the output space where the proxy is high, which need not coincide with regions where true acceptability is high.

A convenient reduced-form way to capture this is to decompose the reward into a true component and a
spurious component.

\begin{assumption}[Misspecified reward]
\label{ass:misspecified}
For each $(x,z)$,
\begin{equation}
r(x,z)=\alpha\, q(x,z) + (1-\alpha)\, s(x,z),
\label{eq:misspecified_reward}
\end{equation}
where $\alpha\in[0,1]$ measures alignment and $s(x,z)$ is a spurious feature that is rewarded in the
preference data but is not intrinsically tied to acceptability in deployment.
\end{assumption}

\begin{proposition}[Goodhart amplification under RLHF]
\label{prop:goodhart}
Under Assumption~\ref{ass:misspecified}, the RLHF policy satisfies
\[
\pi_R(z\mid x)\ \propto\ \pi_0(z\mid x)\exp\!\left(\frac{\alpha q(x,z)+(1-\alpha)s(x,z)}{\lambda}\right).
\]
Fix a prompt $x$. Let $Q_{0,x}$ denote the law of $q(x,Z)$ under $Z\sim \pi_0(\cdot\mid x)$ and let
$Q_{R,x}$ denote the law under $Z\sim \pi_R(\cdot\mid x)$. Then, for all bounded measurable $\varphi$,
\begin{equation}
\mathbb{E}_{Q_{R,x}}[\varphi(Q)]
=
\frac{\mathbb{E}_{Q_{0,x}}\!\left[\varphi(Q)\exp(\alpha Q/\lambda)\,m_{\lambda,x}(Q)\right]}
{\mathbb{E}_{Q_{0,x}}\!\left[\exp(\alpha Q/\lambda)\,m_{\lambda,x}(Q)\right]},
\qquad
m_{\lambda,x}(q):=\mathbb{E}\!\left[\exp((1-\alpha)S/\lambda)\mid Q=q\right],
\label{eq:goodhart_tilt}
\end{equation}
where the conditional expectation defining $m_{\lambda,x}$ is taken under $Z\sim \pi_0(\cdot\mid x)$,
with $Q=q(x,Z)$ and $S=s(x,Z)$.
\end{proposition}

\noindent Proposition~\ref{prop:goodhart} makes the amplification channel explicit. Relative to the
aligned benchmark (where the tilt is solely $\exp(Q/\lambda)$), misspecification introduces an extra
tilt factor $m_{\lambda,x}(Q)$ that depends on how the spurious component varies conditional on true
acceptability (and, in general, on the prompt $x$). As optimisation becomes more aggressive (smaller $\lambda$), this extra factor can
dominate the selection. In particular, if the joint distribution of $(Q,S)$ differs between the
preference-data environment and the deployment environment (a standard distribution shift), then
optimising $r$ can move probability mass toward outputs with high $S$ that do not actually have high
acceptability in deployment. Appendix~\ref{app:rlhf_technical} gives a simple two-point example in
which RLHF lowers the true acceptance rate even though the reward appears predictive in training.

\subsection{Practical implications}
\label{subsec:rlhf_implications}

The RLHF results point to design guidance that mirrors the paper’s earlier embedding-versus-separation
diagnostics, but with the generator now playing the role of the object being embedded. A useful way to
interpret RLHF is as an amortisation choice: it moves what would otherwise be an inference-time
selection problem (sampling multiple candidates and choosing among them) into training. This can be
valuable precisely when selection at deployment is expensive or constrained by latency, compute
budgets, or a limited ability to draw many samples. In those regimes, embedding preferences into the
generator can outperform a nominally superior ex post rule that the user cannot practically execute.

At the same time, when objectives drift over time or vary across users and contexts, the separation
logic becomes more compelling. Theorem~\ref{thm:rlhf_separation} implies a capability tax from
hard-coding a single scalarisation into the generator: as optimisation becomes more aggressive, the
generator concentrates on what is best under the training objective and sacrifices option value under
alternative objectives. A modular architecture that keeps a broadly capable base generator and
implements objective-specific preferences through reranking, thresholds, filters, or routing preserves
that optionality and is therefore more robust to objective movement.

When embedding is unavoidable because post-processing is infeasible, the analysis suggests a more
conservative optimisation stance as a robustness choice. In particular, larger $\lambda$ limits the
degree of concentration of the RLHF policy and therefore limits both the capability tax under
objective mismatch and the amplification of reward misspecification. Finally, it is helpful to be
explicit about which component of the system is being trained. In LLM deployments, the strict
application discussed earlier concerns scorers and routing policies, whereas RLHF modifies the
generator itself. Conflating these levels can obscure available design degrees of freedom: one can
often train scorers in a preference-light or preference-free way while still enforcing stringent
deployment policies via thresholds and routing, reserving generator-level embedding for cases where
selection cannot be done reliably at inference time.

This perspective also clarifies how the RLHF analysis relates to the paper’s main embedding-versus-
separation results. The earlier sections compare predictors under learning frictions, whereas RLHF
changes the generator’s output distribution directly and is not a mean-preserving transformation.
Because RLHF intentionally shifts the distribution of outputs (and thus the mean acceptability level),
Theorem~\ref{thm:rlhf_fixed} naturally takes the form of a first-order dominance result rather than a
convex-order comparison. The analogue of post-processing frictions in the LLM setting is inference-
time search cost: when only a few samples can be drawn under tight latency constraints, embedding can
dominate because the optimal ex post selection rule is out of reach; when post-processing is feasible
and objectives may shift, the separation logic reasserts itself, favouring a rich base distribution
with ex post selection.

\section{Conclusion}\label{sec:conclusion}

This paper studies a simple but pervasive architectural choice in AI pipelines that support
economic decisions: whether to embed downstream preferences into the training objective or
to train a preference-free predictor and implement preferences only through post-processing
and decision rules. The central contribution is to make the “learn then choose” intuition
precise in a setting where learning is endogenous. Modelling training as a choice over
Bayes-plausible posterior distributions connects modern ML practice to information design.
Under a diminishing value of information property, preference embedding flattens the
marginal return to additional informativeness. With a mild comparability condition on
optimal posteriors, this yields a mean-preserving contraction: relative to preference-free
training, preference-embedded training induces a less dispersed distribution of posterior
beliefs. Because indirect utility in expected-utility decision problems is convex in beliefs,
that contraction translates into a robust welfare comparison: for a broad class of downstream
decision problems, the preference-free pipeline weakly dominates, and its welfare is achieved
by applying preferences ex post through the Bayes-optimal decision rule.

The results help reconcile two common pieces of practical advice that otherwise sit in
tension. On the one hand, practitioners often reach for cost-sensitive learning, class
reweighting, or “utility-weighted” objectives to reflect asymmetries in false positives and
false negatives. On the other hand, practitioners also emphasise calibration, proper scoring
rules, and thresholding as deployment tools. The analysis clarifies when the second instinct
should prevail: if one can implement the desired policy reliably at deployment, then training
to learn calibrated probabilities preserves option value across objectives that are uncertain,
contested, or changing. In such environments, embedding one particular set of trade-offs into the loss function can destroy information that cannot be recovered later, even when evaluation ultimately uses the embedded objective.

This separation logic has a natural connection to the algorithmic fairness literature \citep{rambachan2020economic, liang2022algorithmic, liang2024algorithmic}. Fairness criteria are quintessential examples of objectives that vary across jurisdictions,
organisations, and time, and that are frequently revised in response to new evidence or new
normative constraints. Many fairness interventions can be viewed as adjustments to the
decision rule given a score (for example, group-aware thresholds, constrained selection, or
procedural requirements), while the underlying prediction task remains “estimate risk” or
“estimate acceptability.” The paper’s results provide a clean economic rationale for
maintaining a calibrated, preference-free predictive core and implementing fairness as a
transparent post-processing layer that can be audited, updated, and debated without
retraining the predictor. At the same time, the analysis also clarifies why this advice is not
universal: when post-processing is costly or unreliable, upstream embedding can be
welfare-improving by compressing information and reducing the cognitive and institutional
burden of using probabilistic predictions. Understanding when fairness constraints should be
implemented as training-time embedding versus deployment-time policy remains an
important applied design question, and the framework here suggests that the answer should
hinge on the feasibility and cost of downstream implementation as much as on the fairness
criterion itself.

The same themes appear, with higher stakes, in AI alignment and safety. In large language
model deployment, many safety and quality controls already take the form of binary
acceptability scoring and routing: a system assigns a probability that a completion is
acceptable and then decides whether to show, refuse, regenerate, or escalate. For these
modules, the separation principle has a literal interpretation: train scorers with strictly
proper losses to estimate acceptability probabilities, and encode policy trade-offs in
thresholds and routing. More controversially, end-to-end alignment methods such as RLHF
can be interpreted as generator-side preference embedding. When objectives are stable and
the reward is well-specified, embedding can raise the level of acceptability by shifting the
output distribution. But when objectives are multi-dimensional or drift across contexts,
embedding one scalarisation into the generator can impose a capability tax by sacrificing
menu flexibility under alternative criteria, and aggressive optimisation can amplify reward
misspecification in a Goodhart sense. These observations suggest a promising research
direction at the boundary of economics and alignment: characterising the optimal division
of labour between (i) a broadly capable base model, (ii) modular post-processing layers that
encode policy and safety constraints, and (iii) any generator-level fine-tuning that is truly
necessary for latency or usability.

Several extensions would sharpen and broaden the conclusions. One direction is to connect
the framework more directly to fairness constraints by treating fairness objectives as a form
of objective uncertainty and studying which classes of fairness policies can be implemented
as ex post decision rules without sacrificing information. A second direction is to relax the
scalar-belief restriction more fully: multi-class and multi-criteria settings require richer
information orders and raise new selection issues, but they are also where fairness and
alignment concerns typically live (see Online Appendix \ref{app:mv}). A third direction is dynamic: objectives evolve, models are
updated, and governance interventions occur over time. Formalising how modular designs
preserve option value under such evolution, and when the costs of modularity justify
embedding, would make the separation principle operational for real institutions.

Taken together, the results support a disciplined design maxim: when feasible, build general
capabilities that preserve information and apply judgment at the point of choice; embed
preferences only when downstream implementation frictions make separation impractical.

\newpage

\bibliographystyle{aer}
\bibliography{references}

@techreport{caplinMartinMarx2024,
  author       = {Caplin, Andrew and Martin, Daniel J. and Marx, Philip},
  title        = {Modeling Machine Learning: A Cognitive Economic Approach},
  institution  = {National Bureau of Economic Research},
  type         = {Working Paper},
  number       = {30600},
  year         = {2022}
}

@techreport{autorCaplinMartinMarx2025,
  author       = {Autor, David and Caplin, Andrew and Martin, Daniel J. and Marx, Philip},
  title        = {Misaligned by Design: Incentive Failures in Machine Learning},
  institution  = {National Bureau of Economic Research},
  type         = {Working Paper},
  number       = {34504},
  year         = {2025},
  month        = nov
}

@article{strackYang2024,
  author       = {Strack, Philipp and Yang, Kaihao},
  title        = {Privacy-Preserving Signals},
  journal      = {Econometrica},
  volume       = {92},
  number       = {6},
  pages        = {1907--1938},
  year         = {2024},
  month        = nov
}

@incollection{grangerMachina2006,
  author       = {Granger, Clive W. J. and Machina, Mark J.},
  title        = {Forecasting and Decision Theory},
  booktitle    = {Handbook of Economic Forecasting},
  editor       = {Elliott, Graham and Granger, Clive W. J. and Timmermann, Allan},
  volume       = {1},
  pages        = {81--98},
  publisher    = {Elsevier},
  year         = {2006}
}

@article{gneitingRaftery2007,
  author       = {Gneiting, Tilmann and Raftery, Adrian E.},
  title        = {Strictly Proper Scoring Rules, Prediction, and Estimation},
  journal      = {Journal of the American Statistical Association},
  volume       = {102},
  number       = {477},
  pages        = {359--378},
  year         = {2007}
}

@article{kamenicaGentzkow2011,
  author       = {Kamenica, Emir and Gentzkow, Matthew},
  title        = {Bayesian Persuasion},
  journal      = {American Economic Review},
  volume       = {101},
  number       = {6},
  pages        = {2590--2615},
  year         = {2011}
}

@article{sims2003,
  author       = {Sims, Christopher A.},
  title        = {Implications of Rational Inattention},
  journal      = {Journal of Monetary Economics},
  volume       = {50},
  number       = {3},
  pages        = {665--690},
  year         = {2003}
}

@book{gans2025,
    author = {Gans, Joshua S.},
    title = {The Microeconomics of Artificial Intelligence},
    publisher = {MIT Press},
    place = {Cambridge (MA)},
    year = {2025}
}

@article{ouyang2022training,
  title={Training Language Models to Follow Instructions with Human Feedback},
  author={Ouyang, Long and Wu, Jeffrey and Jiang, Xu and Almeida, Diogo and Wainwright, Carroll and Mishkin, Pamela and Zhang, Chong and Agarwal, Sandhini and Slama, Katarina and Ray, Alex and others},
  journal={Advances in Neural Information Processing Systems},
  volume={35},
  pages={27730--27744},
  year={2022}
}

@article{rafailov2023direct,
  title={Direct Preference Optimization: Your Language Model is Secretly a Reward Model},
  author={Rafailov, Rafael and Sharma, Archit and Mitchell, Eric and Ermon, Stefano and Manning, Christopher D and Finn, Chelsea},
  journal={Advances in Neural Information Processing Systems},
  volume={36},
  year={2023}
}

@inproceedings{stiennon2020learning,
  title     = {Learning to Summarize from Human Feedback},
  author    = {Stiennon, Nisan and Ouyang, Long and Wu, Jeffrey and Ziegler, Daniel M. and Lowe, Ryan and Voss, Chelsea and Radford, Alec and Amodei, Dario and Christiano, Paul},
  booktitle = {Advances in Neural Information Processing Systems},
  volume    = {33},
  year      = {2020}
}

@article{ziegler2019finetuning,
  title        = {Fine-Tuning Language Models from Human Preferences},
  author       = {Ziegler, Daniel M. and Stiennon, Nisan and Wu, Jeffrey and Brown, Tom B. and Radford, Alec and Amodei, Dario and Christiano, Paul F. and Irving, Geoffrey},
  journal      = {CoRR},
  volume       = {abs/1909.08593},
  year         = {2019},
  eprint       = {1909.08593},
  archivePrefix= {arXiv},
  primaryClass = {cs.CL}
}

@article{agarwalMoehringRajpurkarSalz2023,
  author = {Agarwal, Nikhil and Moehring, Alex and Rajpurkar, Pranav and Salz, Tobias},
  title = {Combining Human Expertise with Artificial Intelligence: Experimental Evidence from Radiology},
  journal = {NBER Working Paper},
  number = {31422},
  year = {2023},
  institution = {National Bureau of Economic Research}
}

@article{tejedaKumarSmythSteyvers2022,
  author = {Tejeda, Heliodoro and Kumar, Aditya and Smyth, Padhraic and Steyvers, Mark},
  title = {AI-Assisted Decision Making in Healthcare: Effects of Algorithm Advice on Clinical Judgment},
  journal = {Proceedings of the AAAI Conference on Human Computation and Crowdsourcing},
  volume = {10},
  pages = {172--182},
  year = {2022}
}

@book{caplin2025cognitive,
  author = {Caplin, Andrew},
  title = {An Introduction to Cognitive Economics},
  publisher = {Palgrave Macmillan},
  year = {2025}
}

@techreport{caplinDeanLeahy2024,
  title={The ABCs of Who Benefits from Working with AI: Ability, Beliefs, and Calibration},
  author={Caplin, Andrew and Deming, David and Li, Shangwen and Ye, Kadachi Jiada},
  year={2024},
  institution={National Bureau of Economic Research},
  type={Working Paper}
}

@article{varian1990goodness,
  title={Goodness-of-fit in optimizing models},
  author={Varian, Hal R},
  journal={Journal of Econometrics},
  volume={46},
  number={1-2},
  pages={125--140},
  year={1990},
  publisher={Elsevier}
}

@article{zellner1986bayesian,
  title={Bayesian estimation and prediction using asymmetric loss functions},
  author={Zellner, Arnold},
  journal={Journal of the American Statistical Association},
  volume={81},
  number={394},
  pages={446--451},
  year={1986},
  publisher={Taylor \& Francis}
}

@incollection{agrawal2018prediction,
  title={Prediction, judgment, and complexity: a theory of decision-making and artificial intelligence},
  author={Agrawal, Ajay and Gans, Joshua S. and Goldfarb, Avi},
  booktitle={The Economics of Artificial Intelligence: An Agenda},
  pages={89--110},
  year={2018},
  publisher={University of Chicago Press}
}

@book{agrawal2022prediction,
  title={Prediction Machines: The Simple Economics of Artificial Intelligence},
  author={Agrawal, Ajay and Gans, Joshua S. and Goldfarb, Avi},
  year={2022},
  edition = {2nd},
  publisher={Harvard Business Press}
}

@inproceedings{liang2024algorithmic,
  title={Algorithmic Fairness and Social Welfare},
  author={Liang, Annie and Lu, Jay},
  booktitle={AEA Papers and Proceedings},
  volume={114},
  pages={628--632},
  year={2024},
  organization={American Economic Association 2014 Broadway, Suite 305, Nashville, TN 37203}
}

@inproceedings{liang2022algorithmic,
  title={Algorithmic design: Fairness versus accuracy},
  author={Liang, Annie and Lu, Jay and Mu, Xiaosheng},
  booktitle={Proceedings of the 23rd ACM Conference on Economics and Computation},
  pages={58--59},
  year={2022}
}

@article{rambachan2020economic,
  title={An economic perspective on algorithmic fairness},
  author={Rambachan, Ashesh and Kleinberg, Jon and Ludwig, Jens and Mullainathan, Sendhil},
  journal={AEA Papers and Proceedings},
  volume={110},
  pages={91--95},
  year={2020},
  organization={American Economic Association 2014 Broadway, Suite 305, Nashville, TN 37203}
}

\newpage

\appendix

\section{Online Appendix: A multivariate convex-order extension}\label{app:mv}

This appendix sketches how the contraction and separation logic in Theorems~\ref{thm:contraction} and
\ref{thm:separation} extends beyond binary outcomes. The main text focuses on scalar posteriors
because convex order on $[0,1]$ is comparatively tractable and because many economic prediction tasks
are naturally scalar (risk scoring). In multi-class or multi-state environments, the natural belief
object is a \emph{vector} of posterior probabilities. The right analogue of scalar convex order is
the (vector) convex order on the simplex, which continues to deliver robust welfare comparisons for
\emph{all} expected-utility decision problems. The main limitation is comparability: in dimensions
greater than one, convex order is a partial order and ``selection'' issues are more severe.

\subsection{Beliefs and Bayes-plausible posterior distributions}

Let the outcome/state space be finite,
\[
\mathcal{Y} \coloneqq \{1,\dots,K\}, \qquad K\ge 2,
\]
and let $\Delta(\mathcal{Y})$ denote the $K$-simplex:
\[
\Delta(\mathcal{Y}) \coloneqq \left\{q\in\mathbb{R}^K_+:\sum_{y\in\mathcal{Y}} q_y = 1\right\}.
\]
A (vector) posterior belief is $q\in \Delta(\mathcal{Y})$, interpreted as $q_y=\mathbb{P}(Y=y\mid \cdot)$.
Let the prior be $\mu\in\Delta(\mathcal{Y})$.

A signal $S$ induces a random posterior vector
\[
Q \coloneqq \big(\mathbb{P}(Y=y\mid S)\big)_{y\in\mathcal{Y}}\in \Delta(\mathcal{Y}),
\]
and Bayes plausibility implies $\mathbb{E}[Q]=\mu$ (componentwise).

\begin{definition}[Bayes-plausible posterior distributions on the simplex]\label{def:Qmulti}
Let $\mathcal{Q}^{K}(\mu)$ denote the set of random vectors $Q$ taking values in $\Delta(\mathcal{Y})$
such that $\mathbb{E}[Q]=\mu$.
\end{definition}

\subsection{Multivariate convex order}

\begin{definition}[Convex order on $\Delta(\mathcal{Y})$]\label{def:cx_multi}
For $Q,Q'\in\mathcal{Q}^K(\mu)$, write $Q \succeq_{\mathrm{cx}} Q'$ if
\[
\mathbb{E}[\varphi(Q)] \ge \mathbb{E}[\varphi(Q')]
\quad \text{for all convex } \varphi:\Delta(\mathcal{Y})\to\mathbb{R}.
\]
Equivalently, $Q'$ is a mean-preserving contraction of $Q$ in the sense that there exists a coupling
$(\widetilde Q,\widetilde Q')$ with $\widetilde Q\stackrel{d}{=}Q$, $\widetilde Q'\stackrel{d}{=}Q'$,
and $\mathbb{E}[\widetilde Q\mid \widetilde Q']=\widetilde Q'$.
\end{definition}

\noindent As in the scalar case, $Q\succeq_{\mathrm{cx}} Q'$ means that $Q$ is more dispersed (hence
more informative) than $Q'$ while preserving the same mean $\mu$.

\subsection{Decision problems and the value of information with vector beliefs}

A decision problem is a set of feasible actions $\mathcal{A}$ and a payoff function
$u:\mathcal{A}\times\mathcal{Y}\to\mathbb{R}$. Given belief $q\in\Delta(\mathcal{Y})$, the decision
maker's indirect value is
\begin{equation}\label{eq:V_multi}
V(q) \coloneqq \sup_{a\in\mathcal{A}} \sum_{y\in\mathcal{Y}} q_y\,u(a,y).
\end{equation}

\begin{lemma}[Convexity of indirect value on the simplex]\label{lem:convexV_multi}
For any decision problem $(\mathcal{A},u)$, the function $V:\Delta(\mathcal{Y})\to\mathbb{R}$ defined
in \eqref{eq:V_multi} is convex.
\end{lemma}

\begin{proof}
For each fixed action $a$, the map $q\mapsto \sum_{y} q_y u(a,y)$ is linear in $q$. The pointwise
supremum of linear functions is convex.
\end{proof}

\begin{lemma}[More informative posterior distributions improve expected value]\label{lem:infoValue_multi}
If $Q\succeq_{\mathrm{cx}} Q'$, then for any decision problem $(\mathcal{A},u)$,
\[
\mathbb{E}[V(Q)] \ge \mathbb{E}[V(Q')].
\]
\end{lemma}

\begin{proof}
By Lemma~\ref{lem:convexV_multi}, $V$ is convex on $\Delta(\mathcal{Y})$. By Definition~\ref{def:cx_multi},
convex order implies $\mathbb{E}[V(Q)]\ge \mathbb{E}[V(Q')]$.
\end{proof}

\subsection{Multivariate contraction and separation}

To keep the extension parallel to the main text, suppose that training with a $t$-indexed objective
induces a posterior distribution $Q_t\in\mathcal{Q}^K(\mu)$ that solves the reduced-form learning
problem
\begin{equation}\label{eq:learningProblem_multi}
Q_t \in \arg\min_{Q\in\mathcal{Q}^K(\mu)}\left\{\mathbb{E}[H_t(Q)] + C(Q)\right\},
\end{equation}
where $H_t:\Delta(\mathcal{Y})\to\mathbb{R}$ is the Bayes risk associated with the $t$-indexed training
objective and $C$ is a law-invariant learning-friction functional.

The multivariate analogue of Assumption~\ref{ass:DV} is the same ``increasing differences'' condition,
now formulated for convex order on the simplex.

\begin{assumption}[Diminishing value of information on $\Delta(\mathcal{Y})$]\label{ass:DV_multi}
For any $t_1>t_0$ and any $Q,Q'\in\mathcal{Q}^K(\mu)$ with $Q\succeq_{\mathrm{cx}} Q'$,
\begin{equation}\label{eq:increasingDiff_multi}
\mathbb{E}\!\left[ H_{t_1}(Q) - H_{t_0}(Q) \right]
\ge
\mathbb{E}\!\left[ H_{t_1}(Q') - H_{t_0}(Q') \right].
\end{equation}
Moreover, the inequality in \eqref{eq:increasingDiff_multi} is strict whenever
$Q\succeq_{\mathrm{cx}} Q'$ and $Q \not\stackrel{d}{=} Q'$.
\end{assumption}

\noindent A convenient sufficient condition is the same as in Proposition~\ref{prop:sufficientDV}:
if $H_t(q)=H_0(q)+t\,h(q)$ for some convex function $h:\Delta(\mathcal{Y})\to\mathbb{R}$, then
Assumption~\ref{ass:DV_multi} holds (and strictness holds if $h$ is strictly convex). When $H_t$ is
twice differentiable on the relative interior of the simplex, convexity of $h$ can be checked by
requiring its Hessian to be positive semidefinite on the tangent space
$\{v\in\mathbb{R}^K:\sum_{y} v_y=0\}$.

The multivariate analogue of Assumption~\ref{ass:comparable} is where the real pain lives.

\begin{assumption}[Comparability of optimal posteriors]\label{ass:comparable_multi}
For any $t_1>t_0$, and for any minimisers $Q_{t_1}$ and $Q_{t_0}$ of \eqref{eq:learningProblem_multi},
the pair $(Q_{t_0},Q_{t_1})$ is comparable in convex order: either $Q_{t_0}\succeq_{\mathrm{cx}}Q_{t_1}$
or $Q_{t_1}\succeq_{\mathrm{cx}}Q_{t_0}$.
\end{assumption}

\begin{theorem}[Multivariate contraction under preference embedding]\label{thm:contraction_multi}
Suppose Assumptions~\ref{ass:DV_multi} and~\ref{ass:comparable_multi} hold. Let $\{Q_t\}_{t\in[0,1]}$ be
any selection of minimisers of \eqref{eq:learningProblem_multi}. If $t_1>t_0$, then
\[
Q_{t_0} \succeq_{\mathrm{cx}} Q_{t_1}.
\]
\end{theorem}

\begin{proof}
Fix $t_1>t_0$ and let $Q_{t_1}$ and $Q_{t_0}$ be minimisers at $t_1$ and $t_0$, respectively.
Optimality implies
\[
\mathbb{E}[H_{t_1}(Q_{t_1})] + C(Q_{t_1})
\le
\mathbb{E}[H_{t_1}(Q_{t_0})] + C(Q_{t_0}),
\]
and
\[
\mathbb{E}[H_{t_0}(Q_{t_0})] + C(Q_{t_0})
\le
\mathbb{E}[H_{t_0}(Q_{t_1})] + C(Q_{t_1}).
\]
Adding the two inequalities cancels the cost terms and yields
\begin{equation}\label{eq:keyIneq_multi}
\mathbb{E}\!\left[H_{t_1}(Q_{t_1}) - H_{t_0}(Q_{t_1})\right]
\le
\mathbb{E}\!\left[H_{t_1}(Q_{t_0}) - H_{t_0}(Q_{t_0})\right].
\end{equation}
By Assumption~\ref{ass:comparable_multi}, either $Q_{t_0}\succeq_{\mathrm{cx}}Q_{t_1}$ (done) or
$Q_{t_1}\succeq_{\mathrm{cx}}Q_{t_0}$. In the latter case, Assumption~\ref{ass:DV_multi} implies the
reverse inequality to \eqref{eq:keyIneq_multi}, with strictness unless $Q_{t_1}\stackrel{d}{=}Q_{t_0}$.
Hence \eqref{eq:keyIneq_multi} can hold only if $Q_{t_1}\stackrel{d}{=}Q_{t_0}$, in which case
$Q_{t_0}\succeq_{\mathrm{cx}}Q_{t_1}$ holds as well. Therefore $Q_{t_0}\succeq_{\mathrm{cx}}Q_{t_1}$ in
all cases.
\end{proof}

\begin{corollary}[Multivariate robust separation principle]\label{cor:separation_multi}
Under the assumptions of Theorem~\ref{thm:contraction_multi}, let $Q_0$ and $Q_1$ denote the posterior
distributions induced by preference-free ($t=0$) and preference-embedded ($t=1$) training, respectively.
Then, for any decision problem $(\mathcal{A},u)$,
\[
\mathbb{E}[V(Q_0)] \ge \mathbb{E}[V(Q_1)].
\]
\end{corollary}

\begin{proof}
Theorem~\ref{thm:contraction_multi} gives $Q_0\succeq_{\mathrm{cx}} Q_1$. Lemma~\ref{lem:infoValue_multi}
then implies $\mathbb{E}[V(Q_0)]\ge \mathbb{E}[V(Q_1)]$.
\end{proof}

\subsection{Limitations and what this extension does \emph{not} buy you}

This multivariate extension is mathematically clean but economically sobering.

\begin{enumerate}[label=(\roman*), leftmargin=2em]
\item \textbf{Convex order is only a partial order in $K>2$.} In higher dimensions, two Bayes-plausible
posterior distributions are often incomparable. Assumption~\ref{ass:comparable_multi} is therefore much
stronger than its scalar analogue and should be read as a single-index restriction on the feasible
informativeness frontier (optimisers move along a chain of garblings/refinements).

\item \textbf{Checking diminishing value of information is harder.} In one dimension,
Assumption~\ref{ass:DV} reduces to convexity of a scalar increment. On the simplex, it becomes a
matrix (Hessian) restriction along tangent directions. This is still checkable in principle, but it is
not a ``look at the second derivative and smile'' condition.

\item \textbf{Robustness comes at a price.} The conclusion in Corollary~\ref{cor:separation_multi} is
uniform over \emph{all} decision problems because it relies on full convex order. If one is willing to
restrict the class of downstream objectives, weaker information orders can be used. For example, if
downstream value depends only on a one-dimensional index $w\cdot q$ (a scalarisation of beliefs), then
the main-text scalar theory applies to $w\cdot Q$ even when $Q$ itself is not convex-order comparable.

\item \textbf{End-to-end sequence generation is still not ``covered''.} Even though the state space of
a language model can be represented as a huge simplex, the comparability requirements behind
Theorem~\ref{thm:contraction_multi} are typically implausible for such objects. Section~\ref{subsec:llm_binary} discusses the relevant LLM interpretation in the strict binary acceptability setting.
\end{enumerate}

\noindent In short: the multivariate result shows that the separation principle is not inherently
``one-dimensional,'' but it also clarifies why the binary/scalar restriction in the main text is doing
real work. The obstacle is not the convexity of indirect value (which generalises immediately); it is
the ability to order learning outcomes by a tractable informativeness order in dimensions greater than
one.

\newpage

\section{Online Appendix: Technical details for Section~\ref{sec:rlhf}}
\label{app:rlhf_technical}

This appendix collects derivations and proofs omitted from Section~\ref{sec:rlhf}. The main text
keeps the RLHF discussion focused on the design message (embed versus post-process); the results
below are standard but included for completeness and to make the mapping self-contained.

\subsection{A minimal acceptability microfoundation}

Fix a prompt $x$. Let $\theta\in\Theta$ denote a latent ``ground truth'' or user intent with prior
$p(\theta\mid x)$. For each completion $z$, define a binary acceptability indicator
\[
Y=\mathbf{1}\{z \text{ is acceptable given }\theta\}\in\{0,1\}.
\]
The (Bayes) acceptability posterior associated with completion $z$ is
\[
q(x,z)\equiv \mathbb{P}(Y=1\mid x,z)=\mathbb{E}[Y\mid x,z]\in[0,1].
\]
For any generator $\pi(\cdot\mid x)$, define the induced distribution of acceptability posteriors
$Q_{\pi,x}$ as the law of $q(x,Z)$ when $Z\sim \pi(\cdot\mid x)$. Under a prompt distribution $P(x)$,
the unconditional distribution $Q_\pi$ used in Section~\ref{sec:rlhf} is the mixture of these
conditional laws.

A simple downstream decision problem that matches LLM deployment is an accept/regenerate choice.
Let $a\in\{\text{accept},\text{regenerate}\}$, and let payoffs be
\[
u(a,Y)=
\begin{cases}
Y-c(1-Y) & \text{if } a=\text{accept},\\
-d & \text{if } a=\text{regenerate},
\end{cases}
\qquad c>0,\ d\ge 0.
\]
Given belief $q$, the indirect value is
\[
V(q)=\max\{q(1+c)-c,\,-d\},
\]
which is convex and nondecreasing on $[0,1]$. Therefore the expected value of a generator $\pi$ can be
written as
\[
W(\pi)=\mathbb{E}[V(Q_\pi)].
\]
This delivers one concrete microfoundation for the representation used in
Theorem~\ref{thm:rlhf_fixed}. Many other LLM deployment objectives (thresholding, routing, best-of-$n$)
have the same qualitative property: value is nondecreasing in the acceptability posterior.

\subsection{Proof of Proposition~\ref{prop:rlhf_optimal}}

Fix a prompt $x$ and suppress it. Under Assumption~\ref{ass:aligned_reward}, the RLHF problem
\eqref{eq:rlhf_objective} reduces pointwise in $x$ to
\[
\max_{\pi(\cdot)}\ \int q(z)\,\pi(z)\,dz\ -\ \lambda\int \pi(z)\log\!\left(\frac{\pi(z)}{\pi_0(z)}\right)dz
\quad \text{s.t.}\quad \int \pi(z)\,dz=1.
\]
Form the Lagrangian with multiplier $\eta$ on the probability constraint:
\[
\mathcal{L}(\pi,\eta)
=
\int \Big(q(z)\pi(z)-\lambda \pi(z)\log(\pi(z)/\pi_0(z))\Big)\,dz
+\eta\Big(\int \pi(z)\,dz-1\Big).
\]
The first-order condition with respect to $\pi(z)$ yields
\[
q(z)-\lambda\Big(1+\log(\pi(z)/\pi_0(z))\Big)+\eta=0,
\]
so
\[
\pi(z)=\pi_0(z)\exp\!\left(\frac{q(z)+\eta-\lambda}{\lambda}\right)\ \propto\ \pi_0(z)\exp(q(z)/\lambda).
\]
Normalising gives \eqref{eq:rlhf_solution}.

To obtain \eqref{eq:posterior_tilting}, fix any bounded measurable $\varphi:[0,1]\to\mathbb{R}$ and
write $Q=q(Z)$ under $Z\sim\pi_0(\cdot\mid x)$. Let $Q_{0,x}$ denote the law of $Q$ under $Z\sim\pi_0(\cdot\mid x)$.
Using \eqref{eq:rlhf_solution},
\[
\mathbb{E}_{Z\sim\pi_R(\cdot\mid x)}[\varphi(q(Z))]
=
\frac{\mathbb{E}_{Z\sim\pi_0(\cdot\mid x)}\!\left[\varphi(q(Z))\exp(q(Z)/\lambda)\right]}
{\mathbb{E}_{Z\sim\pi_0(\cdot\mid x)}\!\left[\exp(q(Z)/\lambda)\right]}
=
\frac{\mathbb{E}_{Q_{0,x}}\!\left[\varphi(Q)\exp(Q/\lambda)\right]}
{\mathbb{E}_{Q_{0,x}}\!\left[\exp(Q/\lambda)\right]}.
\]
This is \eqref{eq:posterior_tilting}. Integrating the conditional identity over $x\sim P(x)$ yields the
unconditional mixture representation stated in Proposition~\ref{prop:rlhf_optimal}.

\subsection{Proof of Lemma~\ref{lem:tilting_properties}}

Fix a prompt $x$ and suppress it. Let $Q_{0,x}$ and $Q_{R,x}$ denote the induced laws of $q(Z)$ under
$Z\sim \pi_0(\cdot\mid x)$ and $Z\sim \pi_R(\cdot\mid x)$, respectively.

Let $g(Q):=\exp(Q/\lambda)$, which is increasing. For any threshold $t\in[0,1]$, apply
\eqref{eq:posterior_tilting} with $\varphi(Q)=\mathbf{1}\{Q\ge t\}$:
\[
\mathbb{P}_{Q_{R,x}}(Q\ge t)
=
\frac{\mathbb{E}_{Q_{0,x}}\!\left[\mathbf{1}\{Q\ge t\}g(Q)\right]}{\mathbb{E}_{Q_{0,x}}[g(Q)]}.
\]
Because $\mathbf{1}\{Q\ge t\}$ and $g(Q)$ are comonotone, Chebyshev's covariance inequality implies
$\mathbb{E}[\mathbf{1}\{Q\ge t\}g(Q)]\ge \mathbb{E}[\mathbf{1}\{Q\ge t\}]\mathbb{E}[g(Q)]$, hence
$\mathbb{P}_{Q_{R,x}}(Q\ge t)\ge \mathbb{P}_{Q_{0,x}}(Q\ge t)$. This is first-order stochastic
dominance for each fixed $x$. Applying the same identity with $\varphi(Q)=Q$ yields
$\mathbb{E}[Q_{R,x}]\ge \mathbb{E}[Q_{0,x}]$, with strictness unless $Q_{0,x}$ is degenerate (since $g$
is strictly increasing).

The unconditional dominance claim in Lemma~\ref{lem:tilting_properties} follows by mixing over prompts
$x\sim P(x)$, as in the main-text proof.
\subsection{Proof of Proposition~\ref{prop:goodhart}}

Fix $x$ and suppress it. Under Assumption~\ref{ass:misspecified}, the pointwise RLHF solution is
\[
\pi_R(z\mid x)\ \propto\ \pi_0(z\mid x)\exp\!\left(\frac{\alpha q(z)+(1-\alpha)s(z)}{\lambda}\right).
\]
Let $(Q,S)$ denote the pair $(q(Z),s(Z))$ under $Z\sim\pi_0(\cdot\mid x)$ and write $Q_{0,x}$ for the
marginal law of $Q$. For any bounded measurable $\varphi$,
\begin{align*}
\mathbb{E}_{Q_{R,x}}[\varphi(Q)]
&=
\frac{\mathbb{E}_{Z\sim\pi_0(\cdot\mid x)}\!\left[\varphi(q(Z))\exp\!\left(\frac{\alpha q(Z)+(1-\alpha)s(Z)}{\lambda}\right)\right]}
{\mathbb{E}_{Z\sim\pi_0(\cdot\mid x)}\!\left[\exp\!\left(\frac{\alpha q(Z)+(1-\alpha)s(Z)}{\lambda}\right)\right]}\\
&=
\frac{\mathbb{E}_{Q_{0,x}}\!\left[\varphi(Q)\exp(\alpha Q/\lambda)\,\mathbb{E}\!\left[\exp((1-\alpha)S/\lambda)\mid Q\right]\right]}
{\mathbb{E}_{Q_{0,x}}\!\left[\exp(\alpha Q/\lambda)\,\mathbb{E}\!\left[\exp((1-\alpha)S/\lambda)\mid Q\right]\right]},
\end{align*}
which is \eqref{eq:goodhart_tilt} with $m_{\lambda,x}(q)=\mathbb{E}[\exp((1-\alpha)S/\lambda)\mid Q=q]$.
\subsection{A simple Goodhart example: optimisation selects a rare ``hack''}

This example illustrates the sense in which RLHF can amplify misspecification: even if the reward
model looks informative in the baseline distribution, aggressive optimisation can concentrate on
outputs where the proxy is high but true acceptability is low.

Fix a prompt $x$ and consider a baseline generator with finite support
$\mathcal{Z}(x)=\{z_1,z_2,z_3\}$ and probabilities
$\pi_0(z_1)=0.49$, $\pi_0(z_2)=0.49$, $\pi_0(z_3)=0.02$.
Let true acceptabilities and spurious scores be
\[
(q_1,q_2,q_3)=(0.9,\,0.6,\,0.4),
\qquad
(s_1,s_2,s_3)=(0.2,\,0.3,\,1.0),
\]
and let the (misspecified) reward be $r=\alpha q+(1-\alpha)s$ with $\alpha=0.6$. Then
\[
(r_1,r_2,r_3)=(0.62,\,0.48,\,0.64),
\]
so the highest reward is attained by $z_3$ even though $z_3$ has the lowest acceptability among the
two ``reasonable'' outputs $(z_1,z_2)$.

\paragraph{Baseline acceptability.}
Under $\pi_0$, expected acceptability is
\[
\mathbb{E}_{\pi_0}[q(Z)]
=0.49\cdot 0.9+0.49\cdot 0.6+0.02\cdot 0.4
=0.743.
\]
In the baseline distribution, the reward is nonetheless positively correlated with true acceptability:
$z_1$ has both higher $q$ and higher $r$ than $z_2$, and $z_3$ is rare. (Formally, $\text{Cov}(q(Z),r(Z))>0$
under $\pi_0$.)

\paragraph{RLHF selection.}
Under RLHF, $\pi_R^\lambda(z)\propto \pi_0(z)\exp(r(z)/\lambda)$. Since $r_3>r_1>r_2$, we have
$\pi_R^\lambda\Rightarrow \delta_{z_3}$ as $\lambda\downarrow 0$. Hence
\[
\lim_{\lambda\downarrow 0}\ \mathbb{E}_{\pi_R^\lambda}[q(Z)] = q_3=0.4 < 0.743=\mathbb{E}_{\pi_0}[q(Z)].
\]
By continuity, there exists $\bar\lambda>0$ such that $\mathbb{E}_{\pi_R^\lambda}[q(Z)]<
\mathbb{E}_{\pi_0}[q(Z)]$ for all $\lambda<\bar\lambda$.

This is a reduced-form version of a standard ``reward hacking'' logic. The proxy reward appears
useful on the baseline distribution, but the optimisation step makes the generator spend probability
mass on a completion ($z_3$) that is rare under $\pi_0$ and is favoured primarily because it triggers
the spurious component of the reward. The key comparative-static is the same as in the main text:
more aggressive embedding (smaller $\lambda$) increases the scope for such amplification.

\end{document}